\documentclass{llncs}

\usepackage{amsfonts}
\usepackage{color}
\usepackage{tikz}
\usepackage{ltl} 
\usepackage{xspace} 
\usepackage{booktabs,tabularx}
\usepackage{multirow,arydshln}

\newcommand{\todo}[1]{\textcolor{red}{\textbf{TODO:}} #1\\}
\newcommand{\comment}[1]{}

\newcommand{\true}{\mathit{true}}
\newcommand{\false}{\mathit{false}}
\newcommand{\trueval}{\mathsf{true}}
\newcommand{\falseval}{\mathsf{false}}

\newcommand{\nats}{\mathbb{N}}
\newcommand{\natsbot}{\nats \cup \{\bot\}}

\newcommand\defeq{\ensuremath{\mathrel{\raisebox{-.3ex}{$\stackrel{\text{\tiny def}}=$}}}\xspace}

\newcommand{\bigo}[1]{\mathcal{O}(#1)}

\newcommand{\subf}[1]{\mathsf{subf}(#1)}

\newcommand{\trans}{\mathcal{T}}
\newcommand{\Traces}{\mathsf{Traces}}
\newcommand{\val}[2]{\mathit{val}(#1,#2)}

\newcommand{\autA}{\mathcal{A}}
\newcommand{\autB}{\mathcal{B}}
\newcommand{\autN}{\mathcal{N}}
\newcommand{\lang}[1]{\mathcal{L}(#1)}
\newcommand{\relaxfg}{\mathit{Relax}_{\tiny\fg}}

\newcommand{\rej}{\mathsf{rej}}
\newcommand{\Rej}{\mathit{Rej}}

\newcommand{\props}{\mathcal{P}}
\newcommand{\alphabet}{\Sigma}
\newcommand{\ialphabet}{{2^\mathcal{I}}}
\newcommand{\oalphabet}{{2^\mathcal{O}}}
\newcommand{\inpv}{\mathcal{I}}
\newcommand{\outv}{\mathcal{O}}
\newcommand{\inpval}{{\sigma_{I}}}
\newcommand{\outval}{{\sigma_{O}}}

\newcommand{\spec}{{\varphi}}
\newcommand{\softSpec}{{\varphi}}
\newcommand{\fg}{\LTLfinally\LTLglobally}
\newcommand{\gf}{\LTLglobally\LTLfinally}
\newcommand{\gphi}{{\LTLglobally\varphi}}
\newcommand{\gphij}{{\LTLglobally\varphi_j}}
\newcommand{\fgphi}{{\LTLfinally\LTLglobally\varphi}}
\newcommand{\gfphi}{{\LTLglobally\LTLfinally\varphi}}
\newcommand{\Relax}{\mathit{Relax}}

\newcommand{\anot}{\lambda}
\newcommand{\anotfg}{\pi}
\newcommand{\anotb}{\lambda^{\mathbb{B}}}
\newcommand{\anotfgb}{\pi^{\mathbb{B}}}
\newcommand{\anotfgbj}{\pi^{\mathbb{B},j}}
\newcommand{\anotbj}{\lambda^{\mathbb{B},j}}
\newcommand{\anotbjk}{\lambda^{\mathbb{B},j,k}}
\newcommand{\anotbjl}{\lambda^{\mathbb{B},j,l}}
\newcommand{\anotn}{\lambda^{\mathbb{N}}}
\newcommand{\anotfgn}{\pi^{\mathbb{N}}}
\newcommand{\anotfgnj}{\pi^{\mathbb{N},j}}
\newcommand{\anotnj}{\lambda^{\mathbb{N},j}}
\newcommand{\fgvalid}{$\LTLfinally\LTLglobally$--valid}
\newcommand{\succa}{\mathsf{succ}}
\newcommand{\soft}{\mathit{Soft}}

\newtheorem{appxlemma}{Lemma}
\newtheorem{appxprop}{Proposition}
\newtheorem{appxthm}{Theorem}

\newcommand{\supplies}{\mathit{P}}
\newcommand{\loads}{\mathit{L}}
\newcommand{\slp}{s_{l \rightarrow p}}
\newcommand{\caps}{E^+}
\newcommand{\cons}{\mathit{Consumers}}
\newcommand{\supl}{\mathit{Suppliers}}
\newcommand{\office}{\mathit{office}}
\newcommand{\occupied}{\mathit{occupied}}
\newcommand{\passage}{\mathit{passage}}
\newcommand{\specex}{\mathit{special\_exhibition}}
\newcommand{\library}{\mathit{library}}
\newcommand{\ent}{\mathit{entrance}}
\newcommand{\corr}{\mathit{corridor}}
\newcommand{\exh}{\mathit{exhibition}}
\newcommand{\enter}{\mathit{enter}}

\title{Maximum Realizability for Linear Temporal Logic Specifications \vspace{-.25cm}}
\author{Rayna Dimitrova\inst{1}\thanks{This work was done while the author was at The University of Texas at Austin.} \and
Mahsa Ghasemi\inst{2} \and
Ufuk Topcu\inst{2}}
\institute{
$^1$University of Leicester, UK\\
$^2$University of Texas at Austin, USA \vspace{-.25cm}
}

\begin{document}

\maketitle

\begin{abstract}
Automatic synthesis from linear temporal logic (LTL) specifications is widely used in robotic motion planning, control of autonomous systems, and load distribution in power networks. A common specification pattern in such applications consists of an LTL formula describing the requirements on the behaviour of the system, together with a set of additional desirable properties. We study the synthesis problem in settings where the overall specification is unrealizable, more precisely, when some of the desirable properties have to be (temporarily) violated in order to satisfy the system's objective.
We provide a quantitative semantics of sets of safety specifications, and use it to formalize the ``best-effort'' satisfaction of such \emph{soft} specifications while satisfying the \emph{hard} LTL specification. We propose an algorithm for synthesizing implementations that are optimal with respect to this quantitative semantics. Our method builds upon the idea of the bounded synthesis approach, and we develop a MaxSAT encoding which allows for maximizing the quantitative satisfaction of the safety specifications. We evaluate our algorithm on scenarios from robotics and power distribution networks.
\end{abstract}

\section{Introduction}\label{sec:intro}
Automatic synthesis from temporal logic specifications is increasingly becoming a viable alternative for system design in a number of domains such as control and robotics~\cite{FainekosGKP09,Belta16}. The main advantage of synthesis is that it allows the system designer to focus on \emph{what} the system should do, rather than on \emph{how} it should do it. Thus, the main challenge becomes providing the right specification of the system's required behaviour. While significantly easier than developing a system at a lower level, specification design is in its own a difficult and error-prone task. For example, in the case of systems operating in a complex adversarial environment, such as robots, the specification might be over-constrained, and as a result unrealizable, due to failure to account for some of the possible behaviours of the environment. In other cases, the user might have several alternative specifications in mind, possibly with some preferences, and wants to know what the best realizable combination of requirements is. For instance, a temporary violation of a safety requirement might be acceptable, if it is necessary to achieve an important goal. In such cases it is desirable that, when the specification is determined to be unrealizable, the synthesis procedure provides a ``best-effort'' implementation either according to some user-given criteria, or according to the semantics of the specification language.

The challenges of specification design motivate the need to develop synthesis methods for \emph{maximum realizability problem}, where the input to the synthesis tool consists of a \emph{hard specification} which \emph{must} be satisfied by the system, and \emph{soft specifications} which describe other desired, possibly prioritized properties.

A key ingredient of the formulation of maximum realizability problem is a quantitative semantics of the soft requirements. Broadly speaking, one can distinguish between two types of quantitative satisfaction: \emph{intrinsic}, which is based on the semantics of the qualitative operators of the specification language, and \emph{extrinsic}, which requires the user to provide certain quantitative information in terms of costs, weights, priority, or in terms of quantitative operators of the specification language. The approach to maximum realizability that we propose in this paper is applicable to quantitative semantics from both classes.

Our main focus is on soft specifications of the form $\gphi_1,\ldots,\gphi_n$, where each $\softSpec_i$ is a syntactically safe LTL formula. For formulas of the form $\gphi$, we consider a quantitative semantics that is typically used in the context of robustness. More precisely, we consider an intrinsic quantitative semantics which accounts for how often $\softSpec$ is satisfied. In particular, we consider truth values corresponding to $\softSpec$ being satisfied at every point of an execution, being violated only finitely many times, being both violated and satisfied infinitely often, or being continuously violated from some point on. We define a function that determines the value in a given implementation of a conjunction $\gphi_1\wedge\ldots\wedge\gphi_n$ of soft specifications, based on this semantics. Our method then synthesizes an implementation that maximizes the value of the soft specifications. We further extend our proposed method to address quantitative semantics based on user-provided relaxations of the soft specification, and weights capturing their priority.

The approach to maximum realizability that we develop is based on the bounded synthesis technique. Bounded synthesis is able to synthesize implementations of optimal size by leveraging the power of SAT (or QBF, or SMT) solvers. Since maximum realizability is an optimization problem, we reduce its bounded version to maximum satisfiability (MaxSAT). More precisely, we encode the bounded maximum realizability problem with hard and soft specifications as a partial weighted MaxSAT problem, where hard specifications are captured by hard clauses in the MaxSAT formulation, and the weights of soft clauses encode the quantitative semantics of soft specifications. By adjusting these weights our approach can easily capture different variations of quantitative semantics.
Although the formulation encodes the bounded maximum realizability problem (where the maximum size of the implementation is fixed), by providing a bound on the size of the optimal implementation, we are able to establish the completeness of our synthesis method. The existence of such completeness bound is guaranteed by considering quantitative semantics in which the values can be themselves encoded by LTL formulas.

We have applied the proposed synthesis method to examples from two domains where considering combinations of hard and soft specifications is natural and often unavoidable. For example, such a combination of specifications arises in power networks where generators of limited capacity have to power a set of vital and non-vital loads, whose total demand may exceed the capacity of the generators. Another example is robotic navigation, where due to the adversarial nature of the environment in which robots operate, safety requirements might prevent a system from achieving its goal, or a large number of tasks of different nature might not necessarily be consistent when posed together.

{\noindent \bf Related work. }Maximum realizability and several closely related problems have attracted significant attention in recent years. Planning over a finite horizon with prioritized safety requirements was studied in~\cite{TumovaHKFR13}, where the goal is to synthesize a least-violating control strategy. A similar problem for infinite-horizon temporal logic planning was studied in~\cite{KimFS15}, which seeks to revise an inconsistent specification, minimizing the cost of revision with respect to costs for atomic propositions provided by the specifier. \cite{LahijanianAFKV15} describes a method for computing optimal plans for co-safe LTL specifications, where optimality is again with respect to the cost of violating each atomic proposition, which is provided by the user. All of these approaches are developed for the planning setting, where there is no adversarial environment, and thus they are able to reduce the problem to the computation of an optimal path in a graph. The case of probabilistic environments was considered in~\cite{LahijanianK16}. In contrast, in our work we seek to maximize the satisfaction of the given specification against the worst-case behaviour of the environment. 

The problem setting that is the closest to ours is that of~\cite{Tomita2017}. The authors of~\cite{Tomita2017} study a maximum realizability problem in which the specification is a conjunction of a \emph{must} (or \emph{hard}, in our terms) LTL specification, and a number of weighted \emph{desirable} (or \emph{soft}, in our terms) specifications of the form $\gphi$, where $\varphi$ is an arbitrary LTL formula. Their synthesis method requires translating $\gphi$ to a mean-payoff term, by first approximating the LTL formula $\varphi$ with a safety property. This approximation strengthens $\varphi$, while the transition  to a mean-payoff weakens the resulting safety property. Thus,  when $\varphi$ is not a safety property, there is no clear relationship between $\gphi$ and the corresponding mean-payoff term. The mean-payoff terms for the individual desirable specifications are combined in a weighted sum, and the synthesized implementation is optimal with respect to this combined term. In contrast, in our maximum realizability setting each satisfaction value is characterized as an LTL formula which is a relaxation of the original specification, and thus so is the optimal value.  

To the best of our knowledge, our work is the first to employ MaxSAT in the context of reactive synthesis. MaxSAT has been used in~\cite{JumaHM12} for preference-based planning. However, since maximum realizability is concerned with reactive systems, it requires a fundamentally different approach than planning.

The two other main research directions related to maximum realizability are \emph{quantitative synthesis} and \emph{specification debugging}. There are two predominant flavours of quantitative synthesis problems studied in the literature. In the first one (cf.~\cite{BloemCHJ09}), the goal is to generate an implementation that maximizes the value of a mean-payoff objective, while possibly satisfying some $\omega$-regular specification.
In the second setting (cf.~\cite{AlmagorBK16,AlmagorKRV17,TabuadaN16}), the system requirements are formalized in a multi-valued temporal logic. The synthesis methods in these works, however, do not solve directly the corresponding optimization problem, but instead check for the existence of an implementation whose value is in a given set. The optimization problem can then be reduced to a sequence of such queries. 

An optimal synthesis problem for an ordered sequence of prioritized $\omega$-regular properties was studied in~\cite{AlurKW08}, where the classical fixpoint-based game-solving algorithms are extended to a quantitative setting. The main difference in our work is that we allow for incomparable soft specifications each with a number of prioritized relaxations, for which the equivalent sequence of preference-ordered combinations would be of length exponential in the number of soft specification. Our MaxSAT formulation avoids considering explicitly these combinations.

In specification debugging there is a lot of research dedicated to finding good explanations for the unsatisfiability or unrealizability of temporal specifications~\cite{CimattiRST07,Schuppan12,RamanK13}, and more generally at the analysis of specifications~\cite{CimattiRST08,KleinP10,JinDDS13,EhlersR14}. Our approach to maximum realizability can prove useful for specification analysis, since instead of simply providing an optimal value, it computes an optimal relaxation of the given specification in the form of another LTL formula.

\section{Maximum Realizability Problem}\label{sec:problem}
In this section, we first overview linear-time temporal logic, LTL, and the corresponding synthesis problem, which asks to synthesize an implementation, in the form of a transition system, that satisfies an LTL formula given as input. 

Then, we proceed by providing a quantitative semantics for a class of LTL formulas, and the definition of the corresponding maximum realizability problem.

\vspace{-.3cm}
\subsection{Specifications, Transition Systems, and the Synthesis Problem}\label{sec:basic-definitions}

Linear-time temporal logic (LTL) is a standard specification language for formalizing requirements on the behaviour of reactive systems. Given a finite set $\props$ of atomic propositions, the set of LTL formulas is generated by the grammar
$\varphi := p \mid \true \mid \false \mid \neg \varphi \mid \varphi_1 \wedge \varphi_2 \mid \varphi_1 \vee \varphi_2 \mid \LTLnext  \varphi  \mid \varphi_1 \LTLuntil \varphi_2 \mid \varphi_1 \LTLrelease \varphi_2,$
where $p \in \props$  is an atomic proposition, $\LTLnext$ is the \emph{next} operator, $\LTLuntil$ is the \emph{until} operator, and $\LTLrelease$ is the \emph{release} operator. As usual, we  define the derived operators 
\emph{finally}: $\LTLfinally \varphi = \true \LTLuntil \varphi$ and 
\emph{globally}: $\LTLglobally \varphi = \false \LTLrelease \varphi$. An LTL formula $\varphi$ is in negation normal form (NNF) if all the negations appear only in front of atomic propositions. Since every LTL formula can be converted to an equivalent one in NNF, we consider only formulas in NNF.
A \emph{syntactically safe} LTL formula is an LTL formula which contains no occurrences of the $\LTLuntil$ operator in its NNF. 

Let $\Sigma = 2^{\props}$ be the finite alphabet consisting of the valuations of the propositions $\props$. A letter $\sigma \in \Sigma$ is interpreted as the valuation that assigns value $\trueval$ to all $p \in \sigma$ and $\falseval$ to all $p \in \props \setminus \sigma$. LTL formulas are interpreted over infinite words $w \in\alphabet^\omega$. If a word $w \in \alphabet^\omega$ satisfies an LTL formula $\varphi$, we write $w \models \varphi$. The definition of the semantics of LTL can be found for instance in~\cite{BaierKatoen08}. 
We denote with $|\varphi|$ the length of $\varphi$, and with $\subf\varphi$ the set of its subformulas.

In the rest of the paper we assume that the set $\props$ of atomic propositions is partitioned into disjoint sets of \emph{input} propositions $\inpv$ and \emph{output} propositions $\outv$.

A \emph{transition system} over a set of input propositions $\inpv$ and a set of output propositions $\outv$ is a tuple $\trans = (S,s_0,\tau)$, where $S$ is a set of states, $s_0$ is the initial state, and the transition function $\tau : S \times \ialphabet \to S \times \oalphabet$ maps a state $s$ and a valuation $\inpval \in \ialphabet$ of the input propositions  to a successor state $s'$ and a valuation $\outval \in \oalphabet$ to the output propositions. 
Let $\props = \inpv \cup \outv$ be the set of all propositions. For $\sigma \in \Sigma = 2^{\props}$ we denote $\sigma \cap \inpv$ by $\inpval$, and $\sigma \cap \outv$ by $\outval$.

If the set $S$ is finite, then $\trans$ is a finite-state transition system. In this case we define the size $|\trans|$ of $\trans$ to be the number of its states, i.e.,  $|\trans| \defeq |S|$.

An \emph{execution} of $\trans$ is an infinite sequence $s_0, (\inpval_0 \cup \outval_0), s_1, (\inpval_1 \cup \outval_1), s_2\ldots$ such that $s_0$ is the initial state, and $(s_{i+1},\outval_i) = \tau(s_i,\inpval_i)$ for every $i \geq 0$. The corresponding sequence $(\inpval_0 \cup \outval_0),(\inpval_1 \cup \outval_1),\ldots \in \alphabet^\omega$ is called a trace. We denote with $\Traces(\trans)$ the set of all traces of a transition system $\trans$. 

We say that a transition system $\trans$ satisfies an LTL formula $\varphi$ over atomic propositions $\props = \inpv \cup \outv$, denoted $\trans \models \varphi$, if $w \models \varphi$ for every $w \in \Traces(\trans)$.

The \emph{realizability problem for LTL} is to determine whether for a given LTL formula $\varphi$ there exists a transition system $\trans$ that satisfies $\varphi$. The \emph{LTL synthesis problem} asks to construct such a transition system if one exists.

Often, the specification  is a combination of multiple requirements,  which might not be realizable in conjunction. In such a case, in addition to reporting the unrealizability to the system designer, we would like the synthesis procedure to construct an implementation that satisfies the specification ``as much as possible''. Such implementation is particularly useful in the case where some of the requirements  describe desirable but not necessarily essential properties of the system. To determine what ``as much as possible'' formally means, a quantitative semantics of the specification language is necessary. In the next subsection we provide such semantics for a fragment of LTL. The quantitative interpretation is based on the semantics of LTL formulas of the form $\gphi$.

\subsection{Quantitative Semantics of Soft Safety Specifications}\label{sec:quantitative-semantics}
Let $\gphi_1,\ldots,\gphi_n$ be LTL specifications, where each $\softSpec_i$ is a syntactically safe LTL formula. In order to formalize the maximal satisfaction of $\gphi_1\wedge\ldots\wedge\gphi_n$, we first give a quantitative semantics of formulas of the form $\gphi$.

\vspace{-.3cm}
\paragraph{Quantitative semantics of safety specifications.} 
For an LTL formula of the form $\gphi$ and a transition system $\trans$, we define \emph{the value $\val\trans{\gphi}$ of $\gphi$ in $\trans$} as
\[
\val\trans\gphi \defeq
\begin{cases}
(1,1,1) & \text{if } \trans\models\LTLglobally\varphi,\\
(0,1,1) & \text{if } \trans\not\models\LTLglobally\varphi \text{ and } \trans\models\LTLfinally\LTLglobally\varphi,\\
(0,0,1) & \text{if } \trans\not\models\LTLglobally\varphi \text{ and } \trans\not\models\LTLfinally\LTLglobally\varphi, \text{ and }\trans\models\LTLglobally\LTLfinally\varphi,\\
(0,0,0) & \text{if } \trans\not\models\LTLglobally\varphi, \text{ and } \trans\not\models\LTLfinally\LTLglobally\varphi, \text{ and }\trans\not\models\LTLglobally\LTLfinally\varphi.\\
\end{cases}
\]

Thus, the value of $\gphi$ in a transition system $\trans$ is a vector $(v_1,v_2,v_3) \in \{0,1\}^3$, where the value $(1,1,1)$ corresponds to the $\true$ value in the classical semantics of LTL. When $\trans\not\models\gphi$, the values $(0,1,1)$, $(0,0,1)$ and $(0,0,0)$ capture the extent to which $\varphi$ holds or not along the traces of $\trans$. For example, if $\val\trans\gphi =(0,0,1)$, then $\varphi$ holds infinitely often on each trace of $\trans$, but there exists a trace of $\trans$ on which $\varphi$ is violated infinitely often. When $\val\trans\gphi =(0,0,0)$, then on some trace of $\trans$, $\varphi$ holds for at most finitely many positions.

Note that by the definition of $\mathit{val}$, if $\val\trans\gphi = (v_1,v_2,v_3)$, then
\emph{(1)} $v_1=1$ iff $\trans\models \gphi$,
\emph{(2)} $v_2=1$ iff $\trans\models \fgphi$, and
\emph{(3)} $v_3=1$ iff $\trans\models \gfphi$.
Thus, the lexicographic ordering on $\{0,1\}^3$ captures the preference of one transition system over another with respect to the  quantitative satisfaction of $\gphi$.

\begin{example}\label{ex:one-safety}
Consider a robot working as a museum guide. We want to synthesize a transition system representing a navigation strategy for the robot. One of the requirements is that its tour should visit the special exhibition infinitely often, formalized in LTL as $\LTLglobally\LTLfinally \specex$. We also desire that the robot never enters the staff's office, formalized as $\LTLglobally\neg\office$. Now, suppose that initially the key for the special exhibition is in the office. Thus, in order to satisfy $\LTLglobally\LTLfinally \specex$, the robot must violate $\LTLglobally\neg\office$. In any case, a strategy in which the office is entered only once, which satisfies $\fg \neg\office$ is preferable to one which enters the office over and over again, and only satisfies $\gf\neg\office$.
Thus, we want to synthesize a strategy $\trans$ with maximal value $\val\trans{\LTLglobally\neg\office}$.
\end{example}

\vspace{-.2cm}
In order to compare implementations with respect to their satisfaction of a conjunction $\gphi_1 \wedge \ldots \wedge \gphi_n$ of several safety specifications, we will extend the above definition. We consider the case when the specifier has not expressed any preference for the individual conjuncts. Consider the following example.

\vspace{-.2cm}
\begin{example}\label{ex:two-safety} We consider again the museum guide robot, now with two soft safety specifications. The specification $\LTLglobally (\occupied \rightarrow \neg \enter\_{\passage})$ requires that the robot does not enter the narrow passage if it is occupied. The second one, $\LTLglobally \neg\enter\_\library$, requires that the robot never enters the library. Passing through the library is an alternative to using the passage. Now, unless these specifications are given priorities, it is preferable to satisfy each of $(\occupied \rightarrow \neg \enter\_{\passage})$ and $\neg\enter\_\library$ infinitely often, rather than avoid entering the library by going through a occupied passage every time, or vice versa. 
\end{example}

\vspace{-.3cm}
\paragraph{Quantitative semantics of conjunctions.} 
To capture the idea illustrated in Example~\ref{ex:two-safety}, we define a value function, which intuitively 
gives higher values to transition systems in which a fewer number of soft specifications have low values. Formally, let \emph{the value of $\gphi_1 \wedge \ldots \wedge \gphi_n$ in $\trans$} be
\[\val\trans{\gphi_1 \wedge \ldots \wedge \gphi_n} \defeq \big(
\sum_{i=1}^n v_{i,3},
\sum_{i=1}^n v_{i,2},
\sum_{i=1}^n v_{i,1}
\big),\] 
where
$\val\trans{\gphi_i} = (v_{i,1},v_{i,2},v_{i,3})$ for $i \in \{1,\ldots,n\}$. To compare transition systems according to these values, we use lexicographic ordering on $\{0,\ldots,n\}^3$.

\begin{example}\label{ex:two-safety-val}  For the specifications in Example~\ref{ex:two-safety}, the value function defined above assigns value $(2,0,0)$ to a  system that satisfies  $\gf (\occupied \rightarrow \neg \enter\_{\passage})$ and $\gf\neg\enter\_\library$, but satisfies neither of $\fg (\occupied \rightarrow \neg \enter\_{\passage})$ and $\fg\neg\enter\_\library$. It assigns the smaller value $(1,1,1)$
to an implementation that satisfies $\LTLglobally (\occupied \rightarrow \neg \enter_{\passage})$, but not $\gf\neg\enter\_\library$.
\end{example}

Note that in the definition above we have reversed the order of the sums of the values $v_1, v_2$ and $v_3$, with the sum over  $v_3$ being the first, and the sum over $v_1$ being last. In this way, a transition system that satisfies all soft requirements to some extent is considered better in the lexicographic ordering than a transition system that satisfies one of them and violates all the others. We could instead consider the inverse lexicographic ordering, thus giving preference to satisfying some soft specification, over having some lower level of satisfaction over all of them. The next example illustrates the differences between these two variations.

\begin{example}\label{ex:ordering}
For the two soft specifications from Example~\ref{ex:two-safety}, reversing the order of the sums in the definition of $\val\trans{\gphi_1 \wedge \ldots \wedge \gphi_n}$ results in giving the higher value $(1,1,1)$ to a transition system that satisfies $\LTLglobally (\occupied \rightarrow \neg \enter\_{\passage})$ but not $\gf\neg\enter\_\library$, and the lower value $(0,0,2)$ to the one that guarantees only $\gf (\occupied \rightarrow \neg \enter\_{\passage})$ and $\gf\neg\enter\_\library$. The most suitable ordering usually depends on the specific application. 
\end{example}

In Appendix~\ref{sec:generalizations} we discuss generalizations of the framework, where
the user provides a set of relaxations for each of the soft specifications, and possibly a priority ordering among the soft specifications, or numerical weights.

\vspace{-.3cm}
\subsection{Maximum realizability}\label{sec:max-realizability}
Using the definition of quantitative satisfaction of soft safety specifications, we now define the maximum realizability problem, which asks to synthesize a transition system that satisfies a given \emph{hard} LTL specification, and is optimal with respect to the satisfaction of a conjunction of \emph{soft} safety specifications.

{\bf Maximum realizability problem:} Given an LTL formula $\spec$ and formulas $\gphi_1,\ldots,\gphi_n$, where each $\softSpec_i$ is a syntactically safe LTL formula, the maximum realizability problem asks to determine if there exists a transition system $\trans$ such that $\trans \models \spec$, and if the answer is positive, to synthesize a transition system $\trans$ such that $\trans \models \spec$, and such that for every transition system $\trans'$ with $\trans'\models \spec$ it holds that $\val\trans{\gphi_1 \wedge \ldots \wedge \gphi_n} \geq \val{\trans'}{\gphi_1 \wedge \ldots \wedge \gphi_n}$.

{\bf Bounded maximum realizability problem:} Given an LTL formula $\spec$ and formulas $\gphi_1,\ldots,\gphi_n$, where each $\softSpec_i$ is a syntactically safe LTL formula, and a bound $b$, the bounded maximum realizability problem asks to determine if there exists a transition system $\trans$ with $|\trans| \leq b$ such that $\trans \models \spec$, and if the answer is positive, to synthesize a transition system $\trans$ such that $\trans \models \spec$, $|\trans| \leq b$ and such that for every transition system $\trans'$ with $\trans'\models \spec$ and $|\trans'| \leq b$, it holds that $\val\trans{\gphi_1 \wedge \ldots \wedge \gphi_n} \geq \val{\trans'}{\gphi_1 \wedge \ldots \wedge \gphi_n}$.

\section{Preliminaries}
In this section we recall bounded synthesis, introduced in~\cite{ScheweF07a}, and in particular the approach based on reduction to SAT. We begin with the necessary preliminaries from automata theory, and the notion of annotated transition systems.

\subsection{Bounded Synthesis}\label{sec:bounded-synthesis}

A \emph{B\"uchi automaton} over a finite alphabet $\alphabet$ is a tuple $\autA = (Q,q_0,\delta,F)$, where $Q$ is a finite set of states, $q_0$ is the initial state, $\delta \subseteq Q \times \alphabet \times Q$ is the transition relation, and $F \subseteq Q$ is a subset of the set of states. A run of $\autA$ on an infinite word $w=\sigma_0\sigma_1\ldots \in \alphabet^\omega$ is an infinite sequence $q_0,q_1,\ldots$ of states, where $q_0$ is the initial state and for every $i \geq 0$ it holds that $(q_i,\sigma_i,q_{i+1}) \in \delta$. 

A run of a B\"uchi automaton is accepting if it contains infinitely many occurrences of states in $F$. A \emph{co-B\"uchi automaton} $\autA = (Q,q_0,\delta,F)$ differs from a B\"uchi automaton in the accepting condition: a run of a co-B\"uchi automaton is accepting if it contains only \emph{finitely many} occurrences of states in $F$. For a B\"uchi automaton the states in $F$ are called \emph{accepting states}, while for a co-B\"uchi automaton they are called \emph{rejecting states}.
A \emph{nondeterministic} automaton $\autA$ accepts a word $w \in \alphabet^\omega$ if \emph{some} run of $\autA$ on $w$ is accepting.
A \emph{universal} automaton $\autA$ accepts a word $w \in \alphabet^\omega$ if \emph{every} run of $\autA$ on $w$ is accepting.

The \emph{run graph} of a universal automaton $\autA = (Q,q_0,\delta,F)$ on a transition system $\trans = (S,s_0,\tau)$ is the unique graph $G = (V,E)$ with set of nodes $V = S \times Q$ and set of labelled edges $E \subseteq V \times \alphabet \times V$ such that $((s,q),\sigma,(s',q')) \in E$ iff $(q,\sigma,q') \in \delta$ and $\tau(s,\sigma\cap \inpv) = (s',\sigma\cap \outv)$.
That is, $G$ is the product of $\autA$ and $\trans$.

A run graph of a universal B\"uchi (resp.\ co-B\"uchi) automaton is accepting if every infinite path $(s_0,q_0),(s_1,q_1), \ldots$ contains infinitely many (resp. finitely) many occurrences of states $q_i$ in $F$. A transition system $\trans$ is accepted by a universal automaton $\autA$ if the unique run graph of $\autA$ on $\trans$ is accepting. We denote with  $\lang\autA$ the set of transition systems accepted by $\autA$.

The bounded synthesis approach is based on the following property.
\begin{lemma}[\cite{KupfermanV05}]\label{lem:ltl-ucba}
For every LTL formula $\varphi$ we can construct a universal co-B\"uchi automaton $\autA_\varphi$ that has at most $2^{O(|\varphi|)}$ states and is such that for every transition system $\trans$ it holds that $\trans \in \lang{\autA_\varphi}$ if and only if $\trans \models \varphi$.
\end{lemma}

An \emph{annotation} of a transition system $\trans = (S,s_0,\tau)$  with respect to a universal co-B\"uchi automaton $\autA = (Q,q_0,\delta,F)$ is a function $\anot : S \times Q \to \natsbot$ that maps nodes of the run graph of $\autA$ on $\trans$ to the set $\natsbot$. Intuitively, such an annotation is valid if every node $(s,q)$ that is reachable from the node $(s_0,q_0)$ is annotated with a natural number, which is an upper bound on the number of rejecting states on any path from $(s_0,q_0)$ to $(s,q)$. 

Formally, an annotation $\anot : S \times Q \to \natsbot$ is \emph{valid} if
\begin{itemize}
\item $\anot(s_0,q_0) \neq \bot$, i.e., the pair of initial states is labelled with a number, and 
\item whenever $\anot(s,q) \neq \bot$, then for every edge $((s,q),\sigma,(s',q'))$ in the run graph of $\autA$ on $\trans$ we have that $(s',q')$ is annotated with a number (i.e., $\anot(s',q')\neq \bot$), such that
$\anot(s',q') \geq  \anot(s,q)$, and if $q' \in F$, then $\anot(s',q') >  \anot(s,q)$.
\end{itemize}

Valid annotations of finite-state systems correspond to accepting run graphs. An annotation $\anot$ is $c$-bounded if $\anot(s,q) \in \{0,\ldots,c\}\cup\{\bot\}$ for all $s \in S$ and $q \in Q$.\looseness=-1

The synthesis method proposed in \cite{ScheweF07a,FinkbeinerS13} employs the following result in order to 
reduce the bounded synthesis problem to checking the satisfiability of propositional formulas. A transition system $\trans$ is accepted by a universal co-B\"uchi automaton $\autA = (Q,q_0,\delta,F)$ iff there exists a $(|\trans|\cdot|F|)$-bounded valid annotation for $\trans$ and $\autA$. One can estimate a bound on the size of the transition system, which allows to reduce the synthesis problem to its bounded version. Namely, if there exists a transition system that satisfies an LTL formula $\varphi$, then there exists a transition system  satisfying $\varphi$ with  at most $\big(2^{(|\subf\varphi| +\log |\varphi|)}\big)!^2$ states.
\looseness=-1

Let $\autA = (Q,q_0,\delta,F)$ be a universal co-B\"uchi automaton for the LTL formula $\spec$.
Given a bound $b$ on the size of the sought transition system $\trans$, the bounded synthesis problem can be encoded as a satisfiability problem with the following sets of propositional variables and constraints.

{\bf Variables:} The variables represent the sought transition system $\trans$, and the sought valid annotation $\anot$ of the  run graph of $\autA$ on $\trans$. A transition system with $b$ states $S = \{1,\ldots,b\}$ is represented by  Boolean variables $\tau_{s,\inpval,s'}$ and $o_{s,\inpval}$ for every $s, s' \in S$, $\inpval \in \ialphabet$, and output proposition $o \in \outv$. The variable $\tau_{s,\inpval,s'}$ encodes the existence of transition from $s$ to $s'$ on input $\inpval$, and the variable $o_{s,\inpval}$ encodes $o$ being true in the output from state $s$ on input $\inpval$.

The annotation $\anot$ is  represented by the following variables. For each $s\in S$ and $q \in Q$, the annotation is represented by a Boolean variable $\anotb_{s,q}$ and a vector $\anotn_{s,q}$ of $\log(b\cdot |F|)$ Boolean variables: the variable $\anotb_{s,q}$ encodes the reachability of $(s,q)$ from the initial node in the corresponding run graph, and the vector of variables $\anotn_{s,q}$ represents the bound for the node $(s,q)$. 

{\bf Constraints for input-enabled $\trans$}:
$C_\tau \defeq\bigwedge_{s \in S}\bigwedge_{\inpval \in \ialphabet}\bigvee_{s' \in S} \tau_{s,\inpval,s'}$.

{\bf Constraints for valid annotation:}\\
$
\begin{array}{lll}
C_\anot & \defeq & 
\anotb_{s_0,q_0} \wedge\\&& 
\bigwedge_{q,q' \in Q}\bigwedge_{s,s' \in S}\bigwedge_{\inpval \in \ialphabet}
\Big( 
\big(
\anotb_{s,q} \wedge 
\delta_{s,q,\inpval,q'} \wedge 
\tau_{s,\inpval,s'}
\big) \rightarrow 
\succa_\anot(s,q,s',q')
\Big),
\end{array}
$\\
where $\delta_{s,q,\inpval,q'}$ is a formula over the variables $o_{s,\inpval}$ that characterizes the transitions in $\autA$ between  $q$ and $q'$ on labels consistent with $\inpval$, and
$\succa_\anot(s,q,s',q')$ is a formula  over the annotation variables such that
$\succa_\anot(s,q,s',q') \defeq (\anotb_{s',q'} \wedge (\anotn_{s',q'} > \anotn_{s,q}))$ if $q' \in F$, and
$\succa_\anot(s,q,s',q') \defeq (\anotb_{s',q'} \wedge (\anotn_{s',q'} \geq \anotn_{s,q}))$ if $q' \not\in F$.

\subsection{Maximum Satisfiability (MaxSAT)}\label{sec:prelim-maxsat}
While the bounded synthesis problem can be encoded into SAT, for the synthesis of a transition system that satisfies a set of soft specifications as well as possible, we need to solve an optimization problem. In the next section we will reduce the bounded maximum realizability problem to a \emph{partial weighted MaxSAT problem}.

\emph{MaxSAT} is a Boolean optimization problem. Similarly to SAT, instances of MaxSAT are given as propositional formulas in conjunctive normal form (CNF). That is, a MaxSAT instance is a conjunction of clauses, each of which is a disjunction of literals, where a literal is a Boolean variable or its negation. The objective in MaxSAT is to compute a variable assignment that maximizes the number of satisfied clauses. In \emph{weighted MaxSAT}, each clause is associated with a positive numerical weight and the objective is now to maximize the sum of the weights of the satisfied clauses. Finally, in \emph{partial weighted MaxSAT}, there are two types of clauses, namely \emph{hard} and \emph{soft} clauses, where only the soft clauses are associated with weights. In order to be a solution to a partial weighted MaxSAT formula, a variable assignment must satisfy all the hard clauses. An optimal solution additionally maximizes the sum of the weights of the soft clauses.
 
In the encoding in the next section we use hard clauses for the hard specification, and soft clauses to capture the soft specifications in the maximum realizability problem. The weights for the soft clauses  will encode the lexicographic ordering  on values of conjunctions of soft specifications.

\section{From Maximum Realizability to MaxSAT}\label{sec:maxsat-encoding}
We now describe the proposed MaxSAT-based approach to maximum realizability. First, we establish an upper bound on the minimal size of an  implementation that satisfies a given LTL specification $\varphi$ and maximizes the satisfaction of a conjunction of soft safety specifications $\gphi_1, \ldots,\gphi_n$ according to the value function for such formulas described in Section~\ref{sec:quantitative-semantics}.
The established bound can be used to reduce the maximum realizability problem to its bounded version, which, in turn, we encode as a MaxSAT problem.

\subsection{Bounded Maximum Realizability}
To establish an upper bound on the minimal (in terms of size) optimal implementation, we make use of an important property of the function $\mathit{val}$ defined in Section~\ref{sec:quantitative-semantics}. Namely, the property that for each of the possible values of $\gphi_1 \wedge\ldots \wedge\gphi_n$  there is a corresponding LTL formula that encodes this value in the classical LTL semantics, as we formally state in the next lemma.

\begin{lemma}\label{lem:value-as-ltl}
For every transition system $\trans$ and soft safety specifications $\gphi_1,\ldots,$ $\gphi_n$, if $\val\trans{\gphi_1 \wedge \ldots \wedge \gphi_n} = v$, then there exists an LTL formula $\psi_v$ such that $\trans \models \psi_v$ and the following conditions hold
\begin{itemize}
\item[(1)] $\psi_v = \softSpec_1'\wedge\ldots\wedge\softSpec_n'$, where $\softSpec_i' \in\{\gphi_i,\fgphi_i,\gfphi_i,\true\} \text{ for }i=1,\ldots,n$, 
\item[(2)] for every $\trans'$, if $\trans' \models \psi_v$, then $\val{\trans'}{\gphi_1 \wedge \ldots \wedge \gphi_n} \geq v$.
\end{itemize}
\end{lemma}

The following theorem is a consequence of Lemma~\ref{lem:value-as-ltl}. 
\begin{theorem}\label{thm:optimal-bound-safety}
Given an LTL specification $\spec$ and soft safety specifications $\gphi_1,\ldots,$ $\gphi_n$,  
if there exists a transition system $\trans \models \spec$, then there exists  $\trans^*$ such that
\begin{itemize}
\item [(1)] $\val\trans{\gphi_1 \wedge \ldots \wedge \gphi_n} \leq \val{\trans^*}{\gphi_1 \wedge \ldots \wedge \gphi_n}$ for all $\trans$ with $\trans \models \spec$,
\item[(2)] $\trans^* \models \spec$ and $|\trans^*| \leq \left((2^{(b+\log b)})!\right)^2$,
\end{itemize}
 where $b = \max\{|\subf{\spec\wedge\softSpec_1'\wedge\ldots\wedge\softSpec_n'}| \mid \forall i:\ \softSpec_i' \in\{\gphi_i,\fgphi_i,\gfphi_i\}\}$.
\end{theorem}

The bound above is estimated based on the size of the specifications, using a worst-case bound on the size of the corresponding automata. Given automata for all the specifications $\gphi_i,\fgphi_i$ and $\gfphi_i$, a potentially better bound can be estimated based on the size of these automata.

Lemma~\ref{lem:value-as-ltl} immediately provides a naive synthesis procedure, which searches for an optimal implementation by enumerating possible $\psi_v$ formulas and solving the corresponding realizability questions. The total number of these formulas is $4^n$, where $n$ is the number of soft specifications. The approach that we propose avoids this rapid growth, by reducing the optimization problem to a single MaxSAT instance, and thus it also makes use of  the power of MaxSAT solvers.

\subsection{Automata and Annotations for Soft Safety Specifications}\label{sec:automata-safety}
Let $\gphi$ be an LTL formula, where $\varphi$ is a syntactically safe LTL formula.

The first step in the MaxSAT reduction is the construction of a universal B\"uchi automaton for each soft safety specification $\gphi$ and its modification to incorporate the relaxation of $\gphi$ to $\fgphi$, as we now describe.

\begin{proposition}\label{prop:aut-globally}
Given an LTL formula $\gphi$ where $\varphi$ is syntactically safe, we can construct a universal B\"uchi automaton  $\autB_{\scriptsize \gphi} = (Q_{\scriptsize \gphi},q_0^{\scriptsize \gphi},\delta_{\scriptsize \gphi},F_{\scriptsize \gphi})$ such that 
 $\lang{\autB_{\scriptsize \gphi}} = \{\trans \mid \trans \models \gphi\}$, and $\autB_{\scriptsize \gphi}$ has a unique non-accepting sink state, that is, there exists a unique state $\rej_\varphi \in Q_{\scriptsize \gphi}$ such that 
$F_{\scriptsize \gphi} = Q_{\scriptsize \gphi} \setminus \{\rej_\varphi\}$, and for every $\sigma \in \Sigma$ it holds that $\{q \in Q_{\scriptsize \gphi} \mid (\rej_\varphi,\sigma,q) \in \delta_{\scriptsize \gphi}\}  = \{\rej_\varphi\}$.
\end{proposition}

From $\autB_{\scriptsize \gphi}$, which has at most $2^{\mathcal{O}(|\varphi|)}$ states, we obtain a universal automaton $\relaxfg(\gphi)$ constructed by redirecting all the transitions leading  to $\rej_\varphi$ to the initial state $q_0^{\scriptsize \gphi}$. Formally, 
$\relaxfg(\gphi) = (Q,q_0,\delta,F)$, where $Q = Q_{\scriptsize \gphi} \setminus \{\rej_\varphi\}$, $q_0 = q_0^{\scriptsize \gphi}$, $F = F_{\scriptsize \gphi}$ and the transition relation is defined as
\[\delta = \left(\delta_{\scriptsize \gphi} \setminus \{(q,\sigma,q') \in \delta_{\scriptsize \gphi}\mid q' = \rej_\varphi\}\right) \cup \{(q,\sigma,q_0) \mid (q,\sigma,\rej_\varphi)\in \delta_{\scriptsize \gphi}\}.\]

\noindent
Let $\Rej(\relaxfg(\gphi)) = \{(q,\sigma,q_0) \in \delta \mid (q,\sigma,\rej_\varphi) \in \delta_{\scriptsize \gphi}\}$ be the set of transitions in $\relaxfg(\gphi)$ that correspond to transitions in $\autB_{\scriptsize \gphi}$ leading to $\rej_\varphi$. 

The next proposition formalizes the property that a transition system $\trans$ is accepted by $\autB_{\scriptsize \gphi}$ iff the run graph of $\relaxfg(\gphi)$ on $\trans$ does not contain an edge corresponding to a transition in $\Rej(\relaxfg(\gphi))$. 
\begin{proposition}\label{prop:rej-trans}
Let $\trans$ be a transition system and let $G = (V,E)$ be the run graph of $\relaxfg(\gphi)$ on $\trans$. Then, $\trans \in\lang{\autB_{\scriptsize \gphi}}$ iff for every  $((s,q),\sigma,(s',q')) \in E$ with $(q,\sigma,q') \in  \Rej(\relaxfg(\gphi))$, $(s,q)$ is not reachable from $(s_0,q_0)$ in $G$.
\end{proposition}

We define an annotation function $\anotfg : S \times Q \to \natsbot$ for a transition system $\trans = (S,s_0,\tau)$ and the automaton $\relaxfg(\gphi) = (Q,q_0,\delta,F)$. The value of $\anotfg$ for the initial node of the run graph determines whether $\trans$ satisfies $\gphi$ or $\fgphi$.
A function $\anotfg : S \times Q \to \natsbot$ is \emph{\fgvalid}\ annotation if it is such that
\begin{itemize}
\item[\emph{(1)}] $\anotfg(s_0,q_0) \neq \bot$, i.e., the pair of initial states is labelled with a number, and 
\item[\emph{(2)}] if $\anotfg(s,q) \neq \bot$, then for every edge $((s,q),\sigma,(s',q'))$ in the run graph of $\relaxfg(\gphi)$ on $\trans$ we have that $\anotfg(s',q') \neq \bot$, and
\begin{itemize}
\item if $(q,\sigma,q') \in\Rej(\relaxfg(\gphi))$, then $\anotfg(s',q') >  \anotfg(s,q)$, and 
\item if $(q,\sigma,q') \not\in\Rej(\relaxfg(\gphi))$, then $\anotfg(s',q') \geq  \anotfg(s,q)$.
\end{itemize}
\end{itemize}

This definition  guarantees that if $\anotfg$ is a bounded \fgvalid\ annotation for $\trans$ and $\relaxfg(\gphi)$, then the number of rejecting edges in each path in the run graph of $\relaxfg(\gphi)$ on $\trans$ is finite, which implies $\trans \models \fgphi$ as stated below.\looseness=-1

\begin{proposition}\label{prop:anotfg}
Let $\trans = (S,s_0,\tau)$ be a finite-state transition system, and $G = (V,E)$ be the run graph of  $\relaxfg(\gphi)$ on $\trans$. Then, $\trans \models \fgphi$ if and only if there exists a \fgvalid\ $|\trans|$-bounded annotation for $\trans$ and $\relaxfg(\gphi)$. 
\end{proposition}

If the run graph of $\relaxfg(\gphi)$ on $\trans$ contains a reachable edge which belongs to $\relaxfg(\gphi)$, then we can conclude that $\trans\not\models\gphi$. However, if each infinite path in the run graph contains only a finite number of occurrences of edges in $\relaxfg(\gphi)$, then $\trans\models\fgphi$. In particular, we have that if $\anotfg(s_0,q_0) \in \nats$, then $\trans \models \fgphi$, and if $\anotfg$ is $c$-bounded and $\anotfg(s_0,q_0) = c$, then $\trans \models \gphi$.

This property of $\anotfg$ allows us to capture the satisfaction of $\gphi$ and of $\fgphi$ with soft clauses for the same annotation function in the MaxSAT formulation.

\subsection{MaxSAT Encoding of Bounded Maximum Realizability}\label{sec:maxsat}

Let $\autA = (Q,q_0,\delta,F)$ be a universal co-B\"uchi automaton for the LTL formula $\spec$.

For each syntactically safe formula $\gphij$, $j \in\{1,\ldots,n\}$, we consider two universal automata:
the universal automaton $\autB_j = \relaxfg(\gphi_j)= (Q_j,q_0^j,\delta_j,F_j)$ constructed as described in Section~\ref{sec:automata-safety} and 
a universal co-B\"uchi automaton $\autA_j = (\widehat Q_j,\widehat q_0^j,\widehat\delta_j,\widehat F_j)$ for the formula $\gfphi_j$. 
Given a bound $b$ on the size of the sought transition system, we encode the bounded maximum realizability problem as a MaxSAT problem with the following sets of variables and constraints.

{\bf Variables:} The MaxSAT formulation includes the variables from the SAT formulation of the bounded synthesis problem, which represent the sought transition system $\trans$ and the sought valid annotation of the run graph of $\autA$ on $\trans$. Additionally, it includes variables for representing the annotations $\anotfg_j$ and $\anot_j$ for $\autB_j$ and $\autA_j$ respectively, similarly to $\anot$ in the SAT encoding. More precisely, the annotations for $\anotfg_j$ and $\anot_j$ are represented respectively by variables $\anotfgbj_{s,q}$ and $\anotfgnj_{s,q}$  where $s\in S$ and $q \in Q_j$, and $\anotbj_{s,q}$ and $\anotnj_{s,q}$ where $s\in S$ and $q \in \widehat Q_j$.

The set of constraints includes $C_\tau$ and $C_\anot$ from the SAT formulation as hard constraints, as well as the following constraints for the new annotations.

{\bf Hard constraints for valid annotations:}
For each $j=1,\ldots,n$, let
\begin{align*}
C_\anotfg^{j}\defeq\bigwedge_{q,q' \in Q_j}\bigwedge_{s,s' \in S}\bigwedge_{\inpval \in \ialphabet}
\Big( &
\big(
\anotfgbj_{s,q} \wedge
\delta^j_{s,q,\inpval,q'}\wedge 
\tau_{s,\inpval,s'}
\big) \rightarrow 
\succa_{\anotfg}^j(s,q,s',q',\inpval) 
\Big),
\\
C_\anot^{j}\defeq\bigwedge_{q,q' \in \widehat Q_j}\bigwedge_{s,s' \in S}\bigwedge_{\inpval \in \ialphabet}
\Big( &
\big(
\anotbj_{s,q} \wedge
\widehat\delta^j_{s,q,\inpval,q'}\wedge 
\tau_{s,\inpval,s'}
\big) \rightarrow 
\succa_{\anot}^j(s,q,s',q',\inpval)
\Big),
\end{align*}
\begin{align*}
\text{where }\succa_\anotfg^j(s,q,s',q',\inpval) \defeq \anotfgbj_{s',q'} \wedge &
\big(\rej^j(s,q,q',\inpval) \rightarrow \anotfgnj_{s',q'} > \anotfgnj_{s,q}\big) \wedge
\\&
\big(\neg\rej^j(s,q,q',\inpval) \rightarrow \anotfgnj_{s',q'} \geq \anotfgnj_{s,q}\big)
,
 \end{align*}
and $\rej^j(s,q,q',\inpval)$ is a formula over $o_{s,\inpval}$ obtained from $\Rej(\autB_j)$.
The formula $\succa_{\anot}^j(s,\widehat q,s',\widehat q',\inpval)$ is analogous to 
$\succa_{\anot}(s,q,s',q',\inpval)$ defined in Section~\ref{sec:bounded-synthesis}.

{\bf Soft constraints for valid annotations:}
For each $j=1,\ldots,n$ we define
\[
\begin{array}{lllll}
\soft_{\tiny\LTLglobally}^{j} & \quad \defeq \quad &  \anotfgb_{s_0,q_0^j} \wedge (\anotfgn_{s_0,q_0^j} = b) & \qquad & \text{with weight }1,\\
\soft_{\tiny\LTLfinally\LTLglobally}^{j} & \quad \defeq \quad & \anotfgb_{s_0,q_0^j} & \quad & \text{with weight }n, \text{ and }\\
\soft_{\tiny\LTLglobally\LTLfinally}^{j} & \quad \defeq \quad & \anotfgb_{s_0,q_0^j} \vee \anotb_{s_0,\widehat  q_0^j} & \quad & \text{with weight }n^2,\\
\end{array}
\]
where $b \in \nats$ is the bound on the size of the transition system.

The definition of the soft constraints guarantees that $\trans \models \gphi_j$ if and only if there exist corresponding annotations that satisfy all three of the soft constraints for $\gphi_j$. Similarly, if $\trans \models \fgphi_j$, then $\soft_{\tiny\LTLfinally\LTLglobally}^{j}$ and $\soft_{\tiny\LTLglobally\LTLfinally}^{j}$ can be satisfied.

The weights are selected in a way that reflects the ordering of transition systems with respect to their satisfaction of $\gphi_1 \wedge \ldots \wedge \gphi_n$, as stated below.

\begin{lemma}\label{lem:encoding-weights}
Let $\trans'$ and $\trans''$ be transition systems such that $\trans' \models \spec$ and $\trans'' \models \spec$. Let $a'$ and $a''$ be variable assignments satisfying the constraint system, such that $a'$ is an optimal assignment consistent with $\trans'$, and $a''$ is an optimal assignment consistent with $\trans''$. Furthermore, let $w'$ and $w''$ be the sums of the weights of the soft clauses satisfied in $a'$ and $a''$ respectively. Then, it holds that
$\val{\trans'}{\gphi_1 \wedge \ldots \wedge \gphi_n} < \val{\trans''}{\gphi_1 \wedge \ldots \wedge \gphi_n} \text{ iff } w' < w''.$
\end{lemma}

The next results establish the size and correctness of the encoding.

\begin{theorem}\label{thm:encoding-correctness}
Let $\autA$ be a given co-B\"uchi automaton for $\varphi$, and for each $j \in \{1,\ldots,n\}$, let $\autB_j = \relaxfg(\gphi_j)$ be the universal automaton for $\gphi_j$ constructed as in Section~\ref{sec:automata-safety}, and let $\autA_j$ be a universal co-B\"uchi automaton for $\gfphi_j$. 
The constraint system for bound $b \in \nats$ is satisfiable if and only if there exists an implementation $\trans$ with $|\trans| \leq b$  such that $\trans \models \varphi$. Furthermore, from the optimal satisfying assignment to the variables $\tau_{s,\inpval,s'}$ and $o_{s,\inpval}$, one can extract a transition system $\trans^*$ such that for every transition system $\trans$ with $|\trans| \leq b$ and $\trans \models \varphi$ it holds that $\val\trans{\gphi_1 \wedge \ldots \wedge \gphi_n} \leq \val{\trans^*}{\gphi_1 \wedge \ldots \wedge \gphi_n}$.
\end{theorem}

\begin{proposition}\label{prop:encoding-size}
Let $\autA$ be a given co-B\"uchi automaton for $\varphi$, and for each $j \in \{1,\ldots,n\}$, let $\autB_j = \relaxfg(\gphi_j)$ be the universal B\"uchi automaton for $\gphi_j$ constructed as in Section~\ref{sec:automata-safety}, and let $\autA_j$ be a universal co-B\"uchi automaton for $\gfphi$. 
The constraint system for bound $b \in \nats$ has weights in $\mathcal{O}(n^2)$,
\[
\begin{array}{l}
\mathcal{O}\Big(
(b^2 + b \cdot |\outv|)\cdot 2^{|\inpv|} +  
b \cdot |Q|\cdot (1+ \log(b \cdot |Q|)) +\\ 
\phantom{\mathcal{O}(}\sum_{j=1}^n\big(b \cdot |Q_j| (1 + \log(b\cdot |Q_j|))\big)
+ \sum_{j=1}^n\big(b \cdot |\widehat Q_j| (1 + \log(b\cdot |\widehat Q_j|))\big)
\Big)
\end{array}\]
variables, and its size (before conversion to CNF) is
\[
\begin{array}{l}
\mathcal{O}\Big(
|Q|^2 \cdot b^2 \cdot 2^{|\inpv|} \cdot (d + \log(b\cdot |Q|)) + \\
\phantom{\mathcal{O}(}\sum_{j=1}^n\big(|Q_j|^2 \cdot b^2 \cdot 2^{|\inpv|} \cdot (d_j + r_j + \log(b\cdot |Q_j|))\big) +
\\
\phantom{\mathcal{O}(}\sum_{j=1}^n\big(|\widehat Q_j|^2 \cdot b^2 \cdot 2^{|\inpv|} \cdot (\widehat d_j + \log(b\cdot |\widehat Q_j|))\big)
\Big), 
\end{array}\]
\[
\begin{array}{lllllll}
\text{ where} &d &=& \max_{s,q,\inpval,q'}|\delta_{s,q,\inpval,q'}|,&
d_j &=& \max_{s,q,\inpval,q'}|\delta_{s,q,\inpval,q'}^j|,\\
&\widehat d_j &=& \max_{s,q,\inpval,q'}|\widehat \delta_{s,q,\inpval,q'}^j|, \text{ and } &
r_j &=& \max_{s,q,\inpval,q'}|\rej^j(s,q,q',\inpval)|.
\end{array}
\]
\end{proposition}

\section{Experimental Evaluation}\label{sec:experiments}
We implemented the proposed approach to maximum realizability in Python 2.7. For the LTL to automata translation our code calls Spot~\cite{Duret-LutzLFMRX16} version 2.2.4. MaxSAT instances are solved by Open-WBO~\cite{MartinsML14} version 2.0. We evaluated our method on instances of two examples. Each experiment was run on machine with 2.3 GHz Intel Xeon E5-2686 v4 as processor and 16 GiB of memory. While the processor is quad-core, only a single core was used. We set a time-out of 2 hours.

\vspace{-.3cm}
\subsubsection{Robotic Navigation.}
We applied our method to the strategy synthesis for a robotic museum guide. The robot has to give a tour of the exhibitions in a specific order, which constitutes the hard specification. Preferably, it also avoids certain locations, such as the library, or the passage when it is occupied. These preferences are encoded in the soft specifications. There is one input variable that designates the occupancy of the passage, and eight output variables defining the position of the robot. 
The full set of specifications is given in Appendix~\ref{sec:robot-nav-appx}.

Table~\ref{tab:inst-robot} summarizes the results. With implementation bound of 8, the hard specification is realizable and a partial satisfaction of soft specifications is achieved. This strategy always selects the passage to transition from Exhibition 1 to Exhibition 2 and hence, avoids the library. It also violates the requirement of not entering the staff's office once, to acquire access to Exhibition 2. With implementation bound 10, a better strategy (with higher objective value) that alternates between entering the library and the passage is synthesized. This strategy is actually the optimal one with respect to the defined quantitative metric.

\begin{table*}[b!]\centering\footnotesize
	\begin{tabular*}{1\linewidth}{ccccccccccccccc}
		\toprule
		&& \multicolumn{3}{c}{Encoding} && \multicolumn{3}{c}{Solution} && \multicolumn{5}{c}{Time (s)} \\
		\cline{3-5} \cline{7-9} \cline{11-15}
		$|\trans|$ & \phantom{00} & \# vars & \phantom{0} & \# clauses & \phantom{00} & sat. & \phantom{0} & $\Sigma weights$ & \phantom{00} & Spot & \phantom{0} & Open-WBO & \phantom{0} & enc.+solve \\
		\midrule
		2 && 2837 && 18974 && UNSAT && 0 (39) && 0.61 && 0.0088 && 0.10 \\
		4 && 14573 && 98152 && UNSAT && 0 (39) && 0.61 && 0.057 && 0.46 \\
		6 && 40857 && 271474 && UNSAT && 0 (39) && 0.61 && 2.54 && 3.7 \\
		8 && 72057 && 482414 && SAT && 25 (39) && 0.61 && 2853 && 2856 \\
		10 && 135839 && 894892 && SAT && 30 (39) && 0.61 && time-out && time-out \\
		\bottomrule
	\end{tabular*}
	\caption{Results of applying synthesis with maximum realizability to the robotic navigation example, with different bounds on implementation size $|\trans|$. We report on the number of variables and clauses in the encoding, the satisfiability of hard constraints, the value (and bound) of the MaxSAT objective function, the running times of Spot and Open-WBO, and the time of the solver plus the time for generating the encoding.}
	\label{tab:inst-robot}
\end{table*}

\vspace{-.3cm}
\subsubsection{Power Distribution Network.}
We consider the problem of dynamic reconfiguration of power distribution networks. A power network consists of a set $\supplies$ of power supplies (generators) and a set $\loads$ of loads (consumers). The network is a bipartite graph with edges between supplies and loads, where each supply is connected to multiple loads and each load is connected to multiple supplies. Each power supply has an associated capacity, which determines how many loads it can power at a given point in time. Thus, it is possible that not all loads can be powered all the time. Some loads are critical and must be powered continuously, while others are not and should be powered when possible. Some loads can be initializing, meaning they must be powered only initially for several steps. Additionally, power supplies can become faulty during operation, which necessitates dynamic network reconfiguration.

\begin{table*}[b!]\centering\footnotesize

\begin{tabular*}{1\linewidth}{@{}c|ccccccccccccccccccc@{}}
\toprule
&& &\multicolumn{5}{c}{Network} &&\multicolumn{5}{c}{Load characterization} && \multicolumn{5}{c}{Specifications} \\
\cline{4-8} \cline{10-14} \cline{16-20}
& \begin{tabular}{c@{}} Instance \\ \# \end{tabular} & \phantom{ }\phantom{ }\phantom{ } & $|\supplies|$ & \phantom{} & $|\loads|$ & \phantom{} &  $\caps$ & \phantom{ }\phantom{ }\phantom{ } & crit. & \phantom{} & non-crit. & \phantom{} & init. & \phantom{ } \phantom{ }\phantom{ } \phantom{} & $|\inpv|$ & \phantom{} & $|\outv|$ & \phantom{} & \begin{tabular}{@{}c@{}} \# Soft \\ \; spec. \end{tabular}\\
\midrule
fully     & 1 && 3 && 3 &&  1 && 1 && 2 && 0 &&  2 && 9 && 2 \\
connected, & 2 && 3 && 6 &&  2 && 2 && 4 && 0 &&  2 && 18 && 4 \\
switching & 3 && 3 && 3 &&  1 && 0 && 2 && 1 &&  2 && 9 && 2 \\
allowed & 4 && 3 && 6 &&  2 && 1 && 4 && 1 &&  2 && 18 && 4 \\
\midrule
sparse,&5 && 4 && 2 &&  1 && 1 && 1 && 0 &&  3 && 4 && 1 \\
switching&6 && 4 && 4 &&  1 && 1 && 3 && 0 &&  3 && 8 && 3 \\
allowed&7 && 4 && 6 &&  1 && 1 && 5 && 0 &&  3 && 12 && 5 \\
&8 && 4 && 8 &&  1 && 1 && 7 && 0 &&  3 && 16 && 7 \\
\midrule
sparse,&9 && 4 && 2 &&  1 && 1 && 1 && 0 &&  3 && 4 && 5 \\
switching&10 && 4 && 4 &&  1 && 1 && 3 && 0 &&  3 && 8 && 11 \\
restricted&11 && 4 && 6 &&  1 && 1 && 5 && 0 &&  3 && 12 && 17 \\
&12 && 4 && 8 &&  1 && 1 && 7 && 0 &&  3 && 16 && 23 \\
\bottomrule
\end{tabular*}
\caption{Power distribution network instances. 
 An instance is determined by the number supplies $|\supplies|$, the number of loads $|\loads|$, the capacity of supplies $\caps$, the number of critical, non-critical and initializing loads. We also show the number of input $|\inpv|$ and output $|\outv|$ propositions and the number of soft specifications.}
\label{tab:def-inst}
\vspace{-.5cm}
\end{table*}

We apply our method to the problem of synthesizing a relay-switching strategy from LTL specifications. The input propositions $\inpv$ determine which, if any, of the supplies are faulty at each step. We are given an upper bound on the number of supplies that can be simultaneously faulty. The set $\outv$ of output propositions contains one proposition $\slp$ for each load $l \in \loads$  and each supply $p \in \supplies$ that are connected. The meaning of $\slp$ is that $l$ is powered by $p$.

The hard specifications assert that the critical loads must always be powered, the initializing loads should be powered initially, a load is powered by at most one supply, the capacity of supplies is not exceeded, and when a supply is faulty it is not in use.
The soft specifications state that non-critical loads are always powered, and that a powered load should remain powered unless its supply fails.

The full specifications are given in Appendix~\ref{sec:power-net-appx}. Table~\ref{tab:def-inst} describes the instances to which we applied our synthesis method. Power supplies have the same capacity $\caps$ (number of loads they can power) and at most one can be faulty. We consider three categories of instances, depending on the network connectivity (full or sparse), and whether we restrict frequent switching of supplies.
In Table~\ref{tab:inst-stat-1}, we report results for the instances from the first and the last category (for the middle category, see Appendix~\ref{sec:power-net-appx}). In the first set of instances in Table~\ref{tab:inst-stat-1}, we see the effect of specifications with large number of variables (due to full connectivity), where the bottleneck is the LTL to automata translation. In the second set in Table~\ref{tab:inst-stat-1} the limiting factor is the number of soft specifications, leading to large weights and variable numbers in the MaxSAT formulation.

\begin{table*}[!htbp]\centering\footnotesize
	\begin{tabular*}{1\linewidth}{@{}c|cccccccccccccc@{}}
		\toprule
		&& && \multicolumn{3}{c}{Encoding} && \multicolumn{1}{c}{Solution} && \multicolumn{5}{c}{Time (s)} \\
		\cline{5-7} \cline{9-9} \cline{11-15}
		\begin{tabular}{@{}c@{}} Instance \\ \# \end{tabular} & \phantom{0} & $|\trans|$ & \phantom{0} & \# vars & \phantom{0} & \# clauses & \phantom{0} & $\Sigma weights$ & \phantom{0} & Spot & \phantom{0} & Open-WBO & \phantom{0} & enc.+solve \\
		\midrule
		1 && 2 && 278 && 3599 && 11 (14) && 0.14 && 0.0040 && 0.071 \\
		  && 4 && 1102 && 20427 && 11 (14) && 0.14 && 0.017 && 0.14 \\
		  && 6 && 2958 && 60007 && 11 (14) && 0.14 && 0.64 && 0.97 \\
		  && 8 && 4966 && 106495 && 11 (14) && 0.14 && 71 && 71 \\
		\midrule
		2 && 2 && 516 && 14261 && 74 (84) && 85 && 0.013 && 0.16 \\
		&& 4 && 1988 && 81861 && 74 (84) && 85 && 0.079 && 0.59 \\
		&& 6 && 5292 && 240949 && 74 (84) && 85 && 2.2 && 3.6 \\
		&& 8 && 8844 && 427997 && 74 (84) && 85 && 1402 && 1405 \\
		\midrule
		3 && 2 && 334 && 7983 && 11 (14) && 0.29 && 0.032 && 0.12 \\
		&& 4 && 1510 && 45259 && 13 (14) && 0.29 && 0.021 && 0.26 \\
		&& 6 && 4302 && 132295 && 13 (14) && 0.29 && 0.099 && 0.77 \\
		&& 8 && 7334 && 235007 && 13 (14) && 0.29 && 1.0 && 2.2 \\
		\midrule
		4 && 2 && 572 && 38997 && 74 (84) && 200 && 0.023 && 0.37 \\
		&& 4 && 2396 && 224325 && 78 (84) && 200 && 0.14 && 1.5 \\
		&& 6 && 6636 && 659413 && 78 (84) && 200 && 1.1 && 4.9 \\
		&& 8 && 11212 && 1171933 && 78 (84) && 200 && 142 && 150 \\
		\midrule\midrule
		9 && 2 && 1047 && 10328 && 137 (155) && 0.18 && 0.0072 && 0.14 \\
		&& 4 && 3479 && 59298 && 137 (155) && 0.18 && 0.071 && 0.45 \\
		&& 6 && 8559 && 176428 && 137 (155) && 0.18 && 6.0 && 7.1 \\
		&& 8 && 26235 && 610246 && N/A (155) && 0.18 && time-out && time-out \\
		\midrule
		10 && 2 && 2153 && 22714 && 1403 (1463) && 0.48 && 0.015 && 0.27 \\
		&& 4 && 7113 && 130416 && 1403 (1463) && 0.48 && 0.60 && 1.5 \\
		&& 6 && 17505 && 387814 && 1403 (1463) && 0.48 && 716 && 718 \\
		&& 8 && 53693 && 1340908 && N/A (1463) && 0.48 && time-out && time-out \\
		\midrule
		11 && 2 && 3259 && 37596 && 5093 (5219) && 1.9 && 0.023 && 0.42 \\
		&& 4 && 10747 && 216254 && 5093 (5219) && 1.9 && 1.1 && 2.4 \\
		&& 6 && 26451 && 642976 && 5093 (5219) && 1.9 && 801 && 808 \\
		&& 8 && 81151 && 2222770 && N/A (5219) && 1.9 && time-out && time-out \\
		\midrule
		12 && 2 && 4365 && 61342 && 12503 (12719) && 27 && 0.031 && 0.74 \\
		&& 4 && 14381 && 355084 && 12503 (12719) && 27 && 2.6 && 4.8 \\
		&& 6 && 35397 && 1056826 && N/A (12719) && 27 && time-out && time-out \\
		&& 8 && 108609 && 3655032 && N/A (12719) && 27 && time-out && time-out \\
		\bottomrule
	\end{tabular*}
	\caption{Results of applying synthesis with maximum realizability on the instances in Table~\ref{tab:def-inst}, with different bounds on implementation size $|\trans|$. We report on the number of variables and clauses in the encoding, the value (and bound) of the objective function in the MaxSAT instance,  the running times of Spot and Open-WBO, and the time of the solver plus the time for generating the encoding.}
	\label{tab:inst-stat-1}
\end{table*}

\newpage
\bibliographystyle{plain}
\bibliography{main.bib}

\appendix
\newpage
\section{Appendix: Specifications for the Robotic Navigation Example}
\label{sec:robot-nav-appx}
We now consider a scenario in which we want to synthesize a strategy for a robot who is a tour guide for a museum. The map of the museum is shown in Figure. 1. The tour starts at the entrance of the museum where the robot picks up a group of visitors. The main objective is to take the group through the two exhibitions on that floor and then return to the entrance to pick up a new group of people. However, there are further constraints.

\begin{figure*}
	\includegraphics[width=\linewidth]{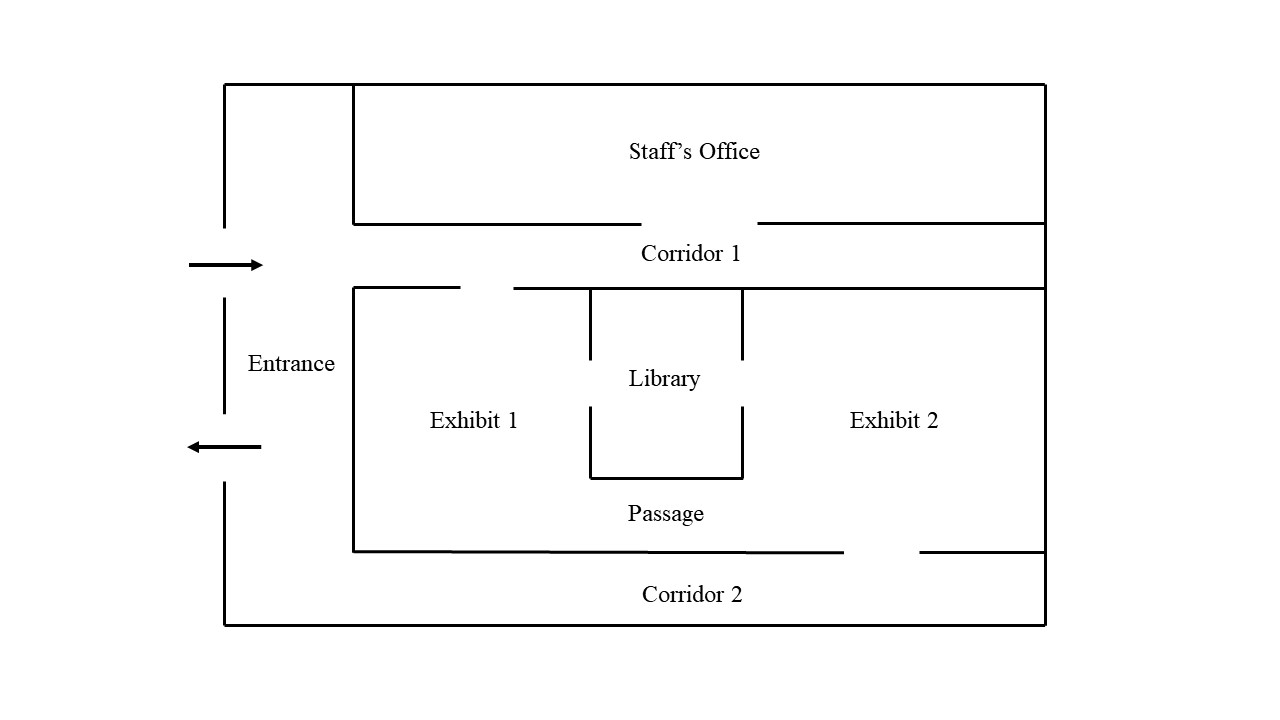}
	\caption{Map of the museum}
\end{figure*}

On one hand, the robot can only gain access to Exhibition 2 by getting a key from the staff's office. 
On the other hand, the robot is asked not to disturb the employees in the office. There is a library Exhibition 1 and Exhibition 2 which can be used to go from one to the other, but it is preferred that visitors do not enter the library. However, when the alternative passage between these two exhibitions is occupied, it is desirable that the robot does not choose to go through there. Clearly, these specifications cannot be realized in conjunction. Given their priorities, we categorize the requirements into hard and soft specifications, and synthesize a strategy which satisfies the hard specifications and maximizes the satisfaction of the soft specifications. We formalize the problem as follows. 

\smallskip
\noindent
{\bf Input propositions:} 
The set $\inpv$ contains a single Boolean variable $\occupied$ that indicates whether the passage between the two Exhibitions is occupied.

\smallskip
\noindent
{\bf Output propositions:}
The set of output propositions $\outv$ consists of eight Boolean variables corresponding to the eight locations on the map: $\ent$, $\corr_1$, $\corr_2$, $\exh_1$, $\exh_2$, $\passage$, $\office$, $\library$.

\smallskip
\noindent
{\bf The hard specification} is the conjunction of the following formulas.

\begin{itemize}
	\item The robot starts at the entrance:
	\begin{equation*}
	\ent
	\end{equation*}
	
	\item At each time step, the robot can occupy only one location:
	\begin{equation*}
	\LTLglobally \bigwedge_{o_1 \in \outv} \left( o_1 \rightarrow \bigwedge_{o_2 \in \outv \backslash \{o_1\}} \neg o_2 \right) .
	\end{equation*}
	
	\item The admissible actions of the robot are to stay in the current location or move to an adjacent one. This leads to eight requirements describing the map. 
	For instance:
	\begin{equation*}
	\LTLglobally \left( \corr_1 \rightarrow \LTLnext \left(\corr_1 \lor \office \lor \exh_1 \right) \right) .
	\end{equation*}
\end{itemize}
    
    {\it Remark: }Due to the requirements above,  the robot will always be in exactly one valid location, i.e., in a transition system that satisfies the specifications it is impossible to reach a state where all output variables are false.

\begin{itemize}
	
	\item The robot must infinitely often visit both exhibitions:
	\begin{equation*}
	\begin{aligned}
	& \LTLglobally \LTLfinally \exh_1 ,\\
	& \LTLglobally \LTLfinally \exh_2 .
	\end{aligned}
	\end{equation*}
	
	\item The robot has to respect the order of visits, by starting from exhibit 1, going to exhibit 2 and finishing at the entrance:
	\begin{equation*}
	\begin{aligned}
	& \LTLglobally (\exh_1 \rightarrow \LTLnext \neg \ent \LTLuntil \exh_2) ,\\
	& \LTLglobally (\exh_2 \rightarrow \LTLnext \neg \exh_1 \LTLuntil \ent) ,\\
	& \LTLglobally (\ent \rightarrow \LTLnext \neg \exh_2 \LTLuntil \exh_1) .
	\end{aligned}
	\end{equation*}
	
	\item The robot does not have access to exhibition 2 before it visits the office:
	\begin{equation*}
	\neg \exh_2 \LTLuntil\office.
	\end{equation*}

\end{itemize}

\smallskip
\noindent
{\bf The set of soft specifications} describes the desirable requirements that the robot does not enter the office, the library, or a occupied passage. Formally:

\begin{itemize}
	\item The robot must not enter the office:
	\begin{equation*}
	\LTLglobally \left(\corr_1 \rightarrow \LTLnext \neg \office \right) .
	\end{equation*}
	
	\item The robot must not enter the library:
	\begin{equation*}
	\LTLglobally \left( \exh_1 \lor \exh_2 \rightarrow \LTLnext \neg \library \right) .
	\end{equation*}
	
	\item The robot must not enter a occupied passage:
	\begin{equation*}
	\LTLglobally \left( \left( \exh_1 \lor \exh_2 \right) \land \LTLnext \occupied \rightarrow \LTLnext \neg \passage \right) .
	\end{equation*}
	
\end{itemize}

\newpage
\section{Appendix: Specifications for the Power Network Example}
\label{sec:power-net-appx}
Consider a power distribution network with set of power supplies $\supplies$ and set of loads $\loads$.
For each supply $p \in \supplies$, we denote with $\caps_p$ the capacity of $p$, that is, how many loads $p$ can power, and with $\cons(p)$ the set of loads connected to $p$ in the network graph. Similarly, for a load $l \in \loads$, let $\supl(l)$ be the set of suppliers to which the load $l$ is connected in the network.
The number of supplies that can be simultaneously faulty is upper bounded by a constant $f$. 

We now give an LTL specification of the relay switching strategy for a given power network.
We begin with the sets of input and output propositions.

\smallskip
\noindent
{\bf Input propositions:} The set $\inpv$ consists of input propositions which form the binary encoding of $f$ integer variables $e_1, e_2, \ldots, e_f$ each with domain $\{0, \ldots, |\loads|\}$. The values of these variables indicate which power supplies are faulty at a given point in time: if $e_i = p$ for some $i$ and $p$, then $p$ is faulty at this time.\footnote{There can be different indices $i$ for which $e_i = p$ at the same time. While this will be redundant, it does not affect the encoding.} 

\smallskip
\noindent
{\bf Output propositions:}
The set of output propositions $\outv$ consists of the Boolean variables $\slp$ for all loads $l \in \loads$ and  supplies $p \in \supplies$ where $l$ and $p$ are connected. The meaning of $\slp$ being true is that $l$ is powered by $p$.

\smallskip
\noindent
{\bf The hard specification} is the conjunction of the following formulas.

\begin{itemize}
\item A critical load must always be powered:
\begin{equation*}
\LTLglobally \left( \bigvee_{p \in \supl(l)} \slp \right) \,, 
\quad \forall \ l \in \loads \ \text{where } l \text{ is critical}.
\end{equation*}

\item An initializing node must be powered during the first two steps:
\begin{equation*}
\left(\bigvee_{p \in \supl(l)} \slp\right) \land\,\,
\LTLnext \left(\bigvee_{p \in \supl(l)} \slp\right)
\,, \quad \forall \ l \in \loads \ \text{where } l \text{ is initializing}.
\end{equation*}

\item A load must only be assigned to one power supply:
\begin{equation*}
\bigwedge_{l \in \loads}\,\,\,
\bigwedge_{p_1 \in \supl(l)}
\LTLglobally \left(s_{l \rightarrow p_1} 
\rightarrow \bigwedge_{p_2 \in \supl(l), p_2 \neq p_1}\neg s_{l \rightarrow p_2}\right).
\end{equation*}

\item The capacity of power supplies must not be exceeded:
\begin{equation*}
\bigwedge_{p \in \supplies} 
\bigwedge_{\substack{{L' \subseteq \cons(p)} \\ {|L'| = \caps_p}}}
\LTLglobally \left( \left(\bigwedge_{l \in L'} s_{l \rightarrow p}\right) \rightarrow
\bigwedge_{l \in \cons(p) \setminus L'} \neg s_{l \rightarrow p} \right).
\end{equation*}

\item When a power supply becomes faulty, no loads can be powered by it:
\begin{equation*}
\bigwedge_{i \in \{1, 2, \ldots f\}}
\bigwedge_{p \in \supplies}
\LTLglobally \left( e_i = p \rightarrow 
\bigwedge_{l \in \cons(p)} \neg \slp \right).
\end{equation*}

\end{itemize}

\smallskip
\noindent
{\bf The set of soft specifications} consist of the requirements for powering the non-vital loads, and optionally, a restriction on switching supplies too often unless they become faulty. The respective formulas are given below.

\begin{itemize}
\item A non-critical load should always be powered:\begin{equation*}
\LTLglobally \left( \bigvee_{p \in \supl(l)} \slp \right) \,, 
\quad \forall \ l \in \loads \ \text{where } l \text{ is non-critical}.
\end{equation*}

\item All loads powered by a supply remain powered by it unless it becomes faulty:
\begin{equation*}
\bigwedge_{l \in \loads}\,\,
\bigwedge_{p \in \supl(l)}
\LTLglobally \left( 
\slp \land \LTLnext \left( \neg 
\bigvee_{i \in \{1, 2, \ldots, f\}} e_i = l \right)
\rightarrow \LTLnext \slp \right). 
\end{equation*}

\end{itemize}

\begin{table*}\centering\footnotesize
	\begin{tabular*}{1\linewidth}{@{}c|cccccccccccccc@{}}
		\toprule
		&& && \multicolumn{3}{c}{Encoding} && \multicolumn{1}{c}{Solution} && \multicolumn{5}{c}{Time (s)} \\
		\cline{5-7} \cline{9-9} \cline{11-15}
		\begin{tabular}{@{}c@{}} Instance \\ \# \end{tabular} & \phantom{00} & $|\trans|$ & \phantom{00} & \# vars & \phantom{0} & \# clauses & \phantom{00} & $\Sigma weights$ & \phantom{00} & Spot & \phantom{0} & Open-WBO & \phantom{0} & enc.+solve \\
		\midrule
		5 && 2 && 235 && 2764 && 3 (3) && 0.53 && 0.044 && 0.079 \\
		&& 4 && 843 && 15726 && 3 (3) && 0.53 && 0.19 && 0.28 \\
		&& 6 && 2115 && 46288 && 3 (3) && 0.53 && 1.8 && 2.0 \\
		&& 8 && 3523 && 82128 && 3 (3) && 0.53 && 3.5 && 4.0 \\
		\midrule
		6 && 2 && 529 && 7586 && 35 (39) && 0.16 && 0.0060 && 0.075 \\
		&& 4 && 1841 && 43272 && 35 (39) && 0.16 && 0.025 && 0.27 \\
		&& 6 && 4617 && 127534 && 35 (39) && 0.16 && 0.64 && 1.4 \\
		&& 8 && 7625 && 226286 && 35 (39) && 0.16 && 71 && 73 \\
		\midrule
		7 && 2 && 823 && 14904 && 137 (155) && 1.3 && 0.011 && 0.14 \\
		&& 4 && 2839 && 85538 && 137 (155) && 1.3 && 0.058 && 0.54 \\
		&& 6 && 7119 && 252556 && 137 (155) && 1.3 && 1.8 && 3.1 \\
		&& 8 && 11727 && 448264 && 137 (155) && 1.3 && 437 && 440 \\
		\midrule
		8 && 2 && 1117 && 31086 && 359 (399) && 26 && 0.017 && 0.28 \\
		&& 4 && 3837 && 180796 && 359 (399) && 26 && 0.13 && 1.1 \\
		&& 6 && 9621 && 536266 && 359 (399) && 26 && 6.1 && 8.9 \\
		&& 8 && 15829 && 952362 && 359 (399) && 26 && 3106 && 3111 \\
		\bottomrule
	\end{tabular*}
	\label{tab:inst-stat-2}
	\caption{Results of applying synthesis with maximum realizability on the instances in Table~\ref{tab:def-inst}, with different bounds on implementation size $|\trans|$. We report on the number of variables and clauses in the encoding, the value (and bound) of the objective function in the MaxSAT instance,  the running times of Spot and Open-WBO, and the time of the solver plus the time for generating the encoding.}	
\end{table*}

\newpage

\newpage
\section{Appendix: Proofs Omitted from Section~\ref{sec:maxsat-encoding}}
\addtocounter{appxlemma}{1}
\begin{appxlemma}\label{lem:value-as-ltl-appx}
For every transition system $\trans$ and soft safety specifications $\gphi_1,\ldots,$ $\gphi_n$, if $\val\trans{\gphi_1 \wedge \ldots \wedge \gphi_n} = v$, then there exists an LTL formula $\psi_v$ such that $\trans \models \psi_v$ and the following conditions hold
\begin{itemize}
\item[(1)] $\psi_v = \softSpec_1'\wedge\ldots\wedge\softSpec_n'$, where $\softSpec_i' \in\{\gphi_i,\fgphi_i,\gfphi_i,\true\} \text{ for }i=1,\ldots,n$, 
\item[(2)] for every $\trans'$, if $\trans' \models \psi_v$, then $\val{\trans'}{\gphi_1 \wedge \ldots \wedge \gphi_n} \geq v$.
\end{itemize}
\end{appxlemma}
\begin{proof}
For each $i \in \{1,\ldots,n\}$, let $(v_{i,1},v_{i,2},v_{i,3}) = \val\trans{\gphi_i}$, and let
\[\psi_v^i \defeq
\begin{cases}
\gphi_i & \text{if } v_{i,1} = 1,\\
\fgphi_i & \text{if } v_{i,1} = 0 \text{ and } v_{i,2} = 1,\\
\gfphi_i  & \text{if } v_{i,2} = 0 \text{ and } v_{i,3} = 1,\\
\true & \text{if } v_{i,3} = 0.\\
\end{cases}
\]

We define $\psi_v  = \bigwedge_{i=1}^n\psi_v^i$. By the definition of $\val\trans{\gphi_i}$ and $\psi_v^i$, we have that $\trans \models \psi_v^i$ for all $i\in \{1,\ldots,n\}$. Thus, we  conclude that $\trans \models \psi_v$. Clearly, $\psi_v$ also satisfies condition \emph{(1)}. Now, consider a transition system $\trans'$ with $\trans' \models \psi_v$.

Let $v = (v_3,v_2,v_1) \in \{0,\ldots,n\}^3$, and $(v_3',v_2',v_1') = \val{\trans'}{\gphi_1 \wedge \ldots \wedge \gphi_n}$. We will show that $v_1' \geq v_1$, $v_2' \geq v_2$ and $v_3' \geq v_3$. Fix some $i \in\{1,\ldots,n\}$.

Let $(v_{i,1}',v_{i,2}',v_{i,3}') = \val{\trans'}{\gphi_i}$. Since $\trans' \models \psi_v$ we have that $\trans' \models \psi_v^i$. Thus by the definition of $\psi_v^i$, we have that if $v_{i,1} = 1$, then $\trans' \models \gphi_i$, and thus $v_{i,1}'=1$. Similarly, if $v_{i,2} = 1$ we can conclude that $v_{i,2}' = 1$, and if  
$v_{i,3} = 1$, then we have $v_{i,3}' = 1$. Since $i \in \{1,\ldots,n\}$ was arbitrary, and since 
\[
\begin{array}{lll}
(v_3,v_2,v_1) & = &  \big(
\sum_{i=1}^n v_{i,3},
\sum_{i=1}^n v_{i,2},
\sum_{i=1}^n v_{i,1}
\big)\text{ and }\\
(v_3',v_2',v_1') & = &  \big(
\sum_{i=1}^n v_{i,3}',
\sum_{i=1}^n v_{i,2}',
\sum_{i=1}^n v_{i,1}'
\big),
\end{array}
\]
we can conclude that $v_1' \geq v_1$, $v_2' \geq v_2$ and $v_3' \geq v_3$. This implies that $(v_3',v_2',v_1') \geq (v_3,v_2,v_1)$ also according to the lexicographic ordering, which proves \emph{(2)}.\qed
\end{proof}

\begin{appxthm}\label{thm:optimal-bound-safety-appx}
Given an LTL specification $\spec$ and soft safety specifications $\gphi_1,\ldots,$ $\gphi_n$,
if there exists a transition system $\trans \models \spec$, then there exists  $\trans^*$ such that
\begin{itemize}
\item [(1)] $\val\trans{\gphi_1 \wedge \ldots \wedge \gphi_n} \leq \val{\trans^*}{\gphi_1 \wedge \ldots \wedge \gphi_n}$ for all $\trans$ with $\trans \models \spec$,
\item[(2)] $\trans^* \models \spec$ and $|\trans^*| \leq \left((2^{(b+\log b)})!\right)^2$,
\end{itemize}
 where $b = \max\{|\subf{\spec\wedge\softSpec_1'\wedge\ldots\wedge\softSpec_n'}| \mid \forall i:\ \softSpec_i' \in\{\gphi_i,\fgphi_i,\gfphi_i\}\}$.
\end{appxthm}
\begin{proof}
Let \[v^* = \max\big\{v \in \{0,\ldots,n\}^3 \mid \exists \trans: \trans\models \varphi \text{ and } \val\trans{\gphi_1 \wedge \ldots \wedge \gphi_n} = v\big\}.\]

Let $\trans$ be a transition system such that $\trans \models \varphi$ and 
$\val\trans{\gphi_1 \wedge \ldots \wedge \gphi_n} = v^*$. According to Lemma~\ref{lem:value-as-ltl}, there exists an LTL formula $\psi_{v^*}$ that satisfies the conditions of the lemma. Thus, $\trans \models \varphi \wedge \psi_{v^*}$. According~\cite{ScheweF07a}, there exists a transition system $\trans^*$ such that $\trans^* \models \varphi \wedge \psi_{v^*}$ and $|\trans^*| \leq \big(2^{|\subf{\varphi \wedge \psi_{v^*}}| +\log |\varphi \wedge \psi_{v^*}|}\big)!^2$. Combining this with the guarantees of Lemma~\ref{lem:value-as-ltl}, we get that $\val{\trans^*}{\gphi_1 \wedge \ldots \wedge \gphi_n} \geq v^*$, and 
$\trans^* \models \spec$ and $|\trans^*| \leq (2^{b+\log b})!^2$.
Thus, $\trans^*$ satisfies condition \emph{(2)}, and since the value $v^*$ is optimal, we have that condition \emph{(1)} holds as well. \qed
\end{proof}

\begin{appxprop}\label{prop:aut-globally-appx}
Given an LTL formula $\gphi$ where $\varphi$ is syntactically safe, we can construct a universal B\"uchi automaton $\autB_{\scriptsize \gphi} = (Q_{\scriptsize \gphi},q_0^{\scriptsize \gphi},\delta_{\scriptsize \gphi},F_{\scriptsize \gphi})$ such that 
 $\lang{\autB_{\scriptsize \gphi}} = \{\trans \mid \trans \models \gphi\}$, and $\autB_{\scriptsize \gphi}$ has a unique non-accepting sink state, that is, there exists a unique state $\rej_\varphi \in Q_{\scriptsize \gphi}$ such that 
$F_{\scriptsize \gphi} = Q_{\scriptsize \gphi} \setminus \{\rej_\varphi\}$, and for every $\sigma \in \Sigma$ it holds that $\{q \in Q_{\scriptsize \gphi} \mid (\rej_\varphi,\sigma,q) \in \delta_{\scriptsize \gphi}\}  = \{\rej_\varphi\}$.
\end{appxprop}
\begin{proof} We first describe the construction of the automaton $\autB_{\scriptsize \gphi}$ of the desired form, and then proceed to prove its correctness.

{\bf Construction.}
We construct a universal B\"uchi automaton $\autB_\varphi$ for the formula $\varphi$ such that, for every word $w \in \alphabet^\omega$, it holds that $w$ is accepted by $\autB_\varphi$ if and only if $w \models \varphi$. To this end, we use the following result from~\cite{KupfermanV01}.
Given a syntactically safe LTL formula $\varphi$, we 
can construct a nondeterministic finite automaton $\autN = (Q_{\autN},q_0^{\autN},\delta_{\autN},F_{\autN})$ with at most $2^{\bigo{|\varphi|}}$ states, and such that:\\
-- if $v \in \alphabet^*$ is accepted by $\autN$, then for all $w' \in \alphabet^\omega$ we have $vw' \not\models \varphi$, and\\
-- for every $w \in \alphabet^\omega$, if $w \not \models \varphi$, then there exists a prefix $v$ of $w$ accepted by $\autN$.
Thus, $\autN$ accepts at least one bad prefix of each word $w \in \alphabet^\omega$ that violates $\varphi$. 

The automaton $\autB_\varphi = (Q_\varphi,q_0^\varphi,\delta_\varphi,F_\varphi)$ is obtained from $\autN$ as follows. The set of states of $\autB_\varphi$ consists of those states of $\autN$ that are not accepting, together with a new state $\rej_\varphi \not\in Q_{\autN}$, that is, $Q_\varphi = (Q_{\autN} \setminus F_{\autN}) \cup \{\rej_\varphi\}$. We let $q_0^\varphi = q_0^{\autN}$ and $F_\varphi = Q_\varphi \setminus \{\rej_\varphi\}$. The transition relation of $\autB_\varphi$ is obtained from $\delta_{\autN}$ by redirecting all transitions leading to states in $F_{\autN}$ to the new state $\rej_\varphi$. Formally, 
\begin{align*}
\delta_\varphi = \left(\delta_{\autN} \cap (Q_\varphi \times \Sigma \times Q_\varphi)\right)& \cup \{(q,\sigma,\rej_\varphi) \mid \sigma \in \alphabet \text{ and }\exists q' \in F_{\autN}.\ (q,\sigma,q') \in \delta_{\autN}\}\\&
\cup \{(\rej_\varphi,\sigma,\rej_\varphi) \mid \sigma \in \alphabet\}.
\end{align*}

\noindent
We now construct a universal B\"uchi automaton $\autB_{\scriptsize \gphi} = (Q_{\scriptsize \gphi},q_0^{\scriptsize \gphi},\delta_{\scriptsize \gphi},F_{\scriptsize \gphi})$ such that $w$ is accepted by $\autB_{\scriptsize \gphi}$ iff $w \models \gphi$.
We let $Q_{\scriptsize \gphi} = Q_\varphi$, $q_0^{\scriptsize \gphi} = q_0^\varphi$, and  $F_{\scriptsize \gphi} = F_\varphi$. The transition relation $\delta_{\scriptsize \gphi}$ extends $\delta_\varphi$ by adding a self-loop at the initial state $q_0^\gphi$ for all transitions from $q_0^\varphi$ in $\delta_\varphi$ that do not lead to $\rej_\varphi$:
\begin{align*}
\delta_{\scriptsize \gphi} = \delta_\varphi & \cup  \{(q_0^\varphi,\sigma,q_0^\varphi) \mid \sigma \in \alphabet \text{ and }\exists q' \in (Q_\varphi \setminus \{\rej_\varphi\}).\ (q_0^\varphi,\sigma,q') \in \delta_\varphi\}.
\end{align*}

{\bf Correctness.}
Let  $\trans \in \lang{\autB_{\scriptsize \gphi}}$. 
Since $\autB$ is a universal B\"uchi automaton, this means that the unique run graph of $\autB$ on $\trans$ is accepting, which in turn means that each infinite path contains infinitely many occurrences of states in $F_{\scriptsize \gphi}$. Since $F_{\scriptsize \gphi}$ contains all states except $\rej_\varphi$, and $\rej_\varphi$ is a sink state, it follows that every infinite path in the run graph contains only states in 
$F_{\scriptsize \gphi}$. 

Suppose, for the sake of contradiction, that $\trans\not\models\varphi$. Thus, there exists $\omega = \sigma_0,\sigma_1,\ldots\in\Traces(\trans)$ such that $\omega\not\models\gphi$. Let $i \geq 0$ be an index such that $\sigma_i,\sigma_{i+1},\ldots \not\models\varphi$. By the choice of the automaton $\autN$, there exists a prefix of $\sigma_i,\sigma_{i+1},\ldots$ accepted by $\autN$.
Since $\omega \in \Traces(\trans)$ and $\trans \in \lang{\autB_{\scriptsize \gphi}}$, every path in the run graph corresponding to $\omega$ never visits $\rej_\varphi$. Thus, since $\delta_{\scriptsize \gphi}$ contains a self-loop at state $q_0^{\scriptsize \gphi}$ with letters not leading to $\rej_\varphi$, there exists a path in the run graph corresponding to $\sigma_0,\ldots,\sigma_{i-1}$ that ends in $q_0^{\scriptsize \gphi}$. By the definition of $\delta_{\varphi}$ and the existence of an accepting run of $\autN$ on a prefix of $\sigma_i,\sigma_{i+1},\ldots$ we can conclude that there exists a path in the run graph of $\autB_{\scriptsize \gphi}$ corresponding to $\omega$ that reaches $\rej_\varphi$, which is a contradiction.

For the other direction, consider a transition system $\trans$ such that $\trans \models \gphi$, and suppose that $\trans \not \in \lang{\autB_{\scriptsize \gphi}}$. This means that there exists an infinite path in the run graph of $\autB_{\scriptsize \gphi}$ on $\trans$ that visits states in $F_{\scriptsize \gphi}$ only finitely many times, which means that this path eventually reaches $\rej_{\varphi}$. Let $\omega = \sigma_0,\sigma_1,\ldots\in\Traces(\trans)$ be the word corresponding to this path, and $i\geq 0$ be the last occurrence of $q_0^{\scriptsize \gphi}$ on this path and $j > i$ be the index of the first occurrence of $\rej_\varphi$.
Due to the definition of $\delta_\varphi$, this implies that there exists an accepting run of $\autN$ on the word $\sigma_i,\ldots,\sigma_{j-1}$. Thus, $\sigma_i,\sigma_{i+1},\ldots\not\models\varphi$, which in turn means  that $\omega\not\models\gphi$. This is a contradiction with $\trans \models \gphi$, and thus we can conclude that $\trans \in \lang{\autB_{\scriptsize \gphi}}$.\qed
\end{proof}

\begin{appxprop}\label{prop:rej-trans-appx}
Let $\trans$ be a transition system and let $G = (V,E)$ be the run graph of $\relaxfg(\gphi)$ on $\trans$. Then, $\trans \in\lang{\autB_{\scriptsize \gphi}}$ iff for every  $((s,q),\sigma,(s',q')) \in E$ with $(q,\sigma,q') \in  \Rej(\relaxfg(\gphi))$, $(s,q)$ is not reachable from $(s_0,q_0)$ in $G$.
\end{appxprop}
\begin{proof}
Suppose, for the sake of contradiction, that $\trans \in\lang{\autB_{\scriptsize \gphi}}$ and let $G'=(V',E')$ be the run graph of $\autB_{\scriptsize \gphi}$ on $\trans$. Suppose that there exists a path $(s_0,q_0),\sigma_0,\ldots (s_l,q_l)$ in $G$ such that there exists an edge $((s_l,q_l),\sigma,(s',q')) \in E$ with $(q,\sigma,q') \in  \Rej(\relaxfg(\gphi))$. Without loss of generality, we assume that 
$(q_i,\sigma_i,q_{i+1}) \not \in \Rej(\relaxfg(\gphi))$ for all $i=0,\ldots,l-1$. Then, the sequence $(s_0,q_0),\sigma_0,\ldots (s_l,q_l),\sigma,(s',q')$ corresponds to a path in the run graph $G'$ of $\autB_{\scriptsize \gphi}$ on $\trans$ which enters the state $\rej_\varphi$. Since $\rej_\varphi$ is a non-accepting sink state, we conclude that $G'$ is not accepting. This implies $\trans \not\in\lang{\autB_{\scriptsize \gphi}}$, which is a contradiction.

Suppose now that for every node $(s,q)$ reachable in $G$ from $(s_0,q_0)$ and every edge $((s,q),\sigma,(s',q')) \in E$ we have that $(q,\sigma,q') \not\in  \Rej(\relaxfg(\gphi))$. 
Assume that $\trans \not\in\lang{\autB_{\scriptsize \gphi}}$, which means that there exists an infinite path from $(s_0,q_0)$ in the run graph of $\autB_{\scriptsize \gphi}$ on $\trans$ that reaches the state $\rej_\varphi$. This path corresponds to a path in $G$ from $(s_0,q_0)$ to some state $(s,q)$ for which there is an edge  $((s,q),\sigma,(s',q')) \in E$ with $(q,\sigma,q') \in  \Rej(\relaxfg(\gphi))$, which is a contradiction.\qed
\end{proof}

\begin{appxprop}\label{prop:anotfg-appx}
Let $\trans = (S,s_0,\tau)$ be a finite-state transition system, and $G = (V,E)$ be the run graph of  $\relaxfg(\gphi)$ on $\trans$. Then, $\trans \models \fgphi$ if and only if there exists a \fgvalid\ $|\trans|$-bounded annotation for $\trans$ and $\relaxfg(\gphi)$. 
\end{appxprop}
\begin{proof}
Suppose that $\trans \models \fgphi$. We will fist show that in every infinite path from $(s_0,q_0)$ in $G$ there are at most $|S|$ occurrences of edges whose corresponding transitions are in $\Rej(\relaxfg(\gphi))$, and then we will use this fact to define a \fgvalid\ $|S|$-bounded annotation. Assume, for the sake of contradiction, that there exists an infinite path $(s_0,q_0),\sigma_0,(s_1,q_1),\sigma_1,\ldots$ such that for infinitely many positions $i\geq 0$ it holds that $(q_i,\sigma_i,q_{i+1}) \in \Rej(\relaxfg(\gphi))$. Let $i_1 < i_2 <\ldots$ be a sequence of such positions. By the construction of $\relaxfg(\gphi)$, we have $q_{i_j+1} = q_0$ for each $i_j$. Thus, using reasoning similar to that in Proposition~\ref{prop:rej-trans}, we can show that the trace $\sigma_0,\sigma_1,$ contains infinitely many positions $k$ such that $\sigma_k,\sigma_{k+1},\ldots \not\models\varphi$. This means that $\sigma_0,\sigma_1,\ldots \not \models \fgphi$. Since $\sigma_0,\sigma_1,\ldots \in \Traces(\trans)$, we can conclude that $\trans \not \models \fgphi$, which is a contradiction. 

Since the number of distinct nodes in $G$ of the form $(s,q_0)$ is $|S|$, we obtain an upper bound of $|S|$ occurrences of transitions from $\Rej(\relaxfg(\gphi))$ on every path in $G$. Thus, we can construct a \fgvalid\ $|S|$-bounded annotation $\anotfg$ by mapping each reachable node $(s,q)$ to the maximal number of transitions from $\Rej(\relaxfg(\gphi))$ on a path from $(s_0,q_0)$ to $(s,q)$, and mapping each unreachable node to $\bot$.

For the other direction, suppose that $\anotfg$ is a \fgvalid\ $|S|$-bounded annotation for $\trans$ and $\relaxfg(\gphi)$. Assume that $\trans \not\models \fgphi$. This means that there exists a trace $w = \sigma_0,\sigma_1,\ldots \in \Traces(\trans)$ such that for every position $i$ it holds that $\sigma_i,\sigma_i+1,\ldots \not\models \gphi$. Let $s_0,\sigma_0,s_1,\sigma_1\ldots$ be the execution of $\trans$ corresponding to $w$. Since $\sigma_i,\sigma_i+1,\ldots \not\models \gphi$, with reasoning similar to Proposition~\ref{prop:rej-trans} we can establish that there exists a path in
$G$ starting from $(s_i,q_0)$ that eventually takes an edge corresponding to a transition in $\Rej(\relaxfg(\gphi))$, and by the construction of $\relaxfg(\gphi)$, this transition leads to the node $(s_j,q_0)$ for some $j > i$. Thus, by induction, we can establish the existence of an infinite path $(s_0,q_0),(s_1,q_1),\ldots$ in $G$ that contains infinitely many occurrences of edges whose transitions are in $\Rej(\relaxfg(\gphi))$. Since the annotation $\anotfg$ is \fgvalid, we can show by induction that for each $i\geq 0$ it holds that $\anotfg(s_i,q_i) \in \nats$ and $\anotfg(s_{i+1},q_{i+1})  \geq \anotfg(s_i,q_i)$. Since $G$ is finite, this path contains an edge  $((s_i,q_i),\sigma_i,(s_{i+1},q_{i+1}))$ for which $(q_i,\sigma_i,q_{i+1})\in\Rej(\relaxfg(\gphi))$, and which is such that there exists $j \leq i$ such that $(s_{i+1},q_{i+1}) = (s_j,q_j)$. Since the annotation $\anotfg$ is \fgvalid, we have that $\anotfg(s_j,q_j) \leq \anotfg(s_i,q_i)$ and $\anotfg(s_i,q_i) < \anotfg(s_{i+1},q_{i+1})$, which contradicts $(s_{i+1},q_{i+1}) = (s_j,q_j)$. Thus, by contradiction, we conclude that $\trans \models \fgphi$. \qed
\end{proof}

\begin{appxlemma}\label{lem:encoding-weights-appx}
Let $\trans'$ and $\trans''$ be transition systems such that $\trans' \models \spec$ and $\trans'' \models \spec$. Let $a'$ and $a''$ be variable assignments satisfying the constraint system, such that $a'$ is an optimal assignment consistent with $\trans'$, and $a''$ is an optimal assignment consistent with $\trans''$. Furthermore, let $w'$ and $w''$ be the sums of the weights of the soft clauses satisfied in $a'$ and $a''$ respectively. Then, it holds that
$\val{\trans'}{\gphi_1 \wedge \ldots \wedge \gphi_n} < \val{\trans''}{\gphi_1 \wedge \ldots \wedge \gphi_n} \text{ iff } w' < w''.$
\end{appxlemma}
\begin{proof}
Let $(v_3',v_2',v_1')  = \val{\trans'}{\gphi_1 \wedge \ldots \wedge \gphi_n}$ and 
$(v_3'',v_2'',v_1'')  = \val{\trans''}{\gphi_1 \wedge \ldots \wedge \gphi_n}$. This means that there are exactly $v_1'$ distinct indices $i \in \{1,\ldots,n\}$ such that $\trans' \models \gphi_i$,  $v_2'$ distinct indices $i \in \{1,\ldots,n\}$ such that $\trans' \models \fgphi_i$ and $v_3'$ distinct indices $i \in \{1,\ldots,n\}$ such that $\trans' \models \gfphi_i$. Since $a'$ is an optimal satisfying assignment corresponding to $\trans'$, we have that $a'$ satisfies exactly $v_1'$ of the soft clauses $\soft_{\tiny\LTLglobally}^{j}$, exactly $v_2'$ of the soft clauses $\soft_{\tiny\LTLfinally\LTLglobally}^{j}$ and exactly $v_3'$ of the soft clauses $\soft_{\tiny\LTLglobally\LTLfinally}^{j}$. This means that $w' = v_1' + v_2' \cdot n + v_3' \cdot n^2$. In a similar way we can conclude that $w'' = v_1'' + v_2'' \cdot n + v_3'' \cdot n^2$ holds for $\trans''$.

First, suppose that $(v_3',v_2',v_1') < (v_3'',v_2'',v_1'')$. There are three possible cases:

\noindent
\emph{Case 1:} $v_3' = v_3''$, $v_2' = v_2''$ and $v_1' < v_1''$. Then, $w'' - w' = (v_1'' - v_1') > 0$.

\noindent
\emph{Case 2:} $v_3' = v_3''$ and $v_2' < v_2''$. Then, $w'' - w' = (v_2'' - v_2')\cdot n  + (v_1'' - v_1')$.
Since $\trans' \models \gphi_i$ implies $\trans' \models \fgphi_i$, we have that $v_1' - v_1'' \leq n-1$, due to the fact that $v_2'' - v_2' \geq 1$. Thus, we conclude $w'' - w' \geq n - (n-1) = 1 >0$.

\noindent
\emph{Case 3:} $v_3' < v_3''$. Now, $w'' - w' = (v_3'' - v_3')\cdot n^2 + (v_2'' - v_2')\cdot n  + (v_1'' - v_1')$. Again, since $\trans' \models \gphi_i$ implies $\trans' \models \gfphi_i$ and 
$\trans' \models \fgphi_i$ implies $\trans' \models \gfphi_i$, we have that $v_1' - v_1'' \leq n-1$ and 
$v_2' - v_2'' \leq n-1$, both due to the fact that $v_3'' - v_3' \geq 1$. Thus, we conclude $w'' - w' \geq n^2 - (n-1)\cdot n -(n-1) = 1 >0$.

In all three cases we showed that $w' < w''$.

For the other direction, suppose that $w' < w''$. If we assume that $(v_3',v_2',v_1') \geq (v_3'',v_2'',v_1'')$, then we can show as above that $w'' \geq w'$, which is a contradiction. Hence, we have that $(v_3',v_2',v_1') < (v_3'',v_2'',v_1'')$, which concludes the proof.\qed
\end{proof}

\begin{appxthm}\label{thm:encoding-correctness-appx}
Let $\autA$ be a given co-B\"uchi automaton for $\varphi$, and for each $j \in \{1,\ldots,n\}$, let $\autB_j = \relaxfg(\gphi_j)$ be the universal automaton for $\gphi_j$ constructed as in Section~\ref{sec:automata-safety}, and let $\autA_j$ be a universal co-B\"uchi automaton for $\gfphi_j$. 
The constraint system for bound $b \in \nats$ is satisfiable if and only if there exists an implementation $\trans$ with $|\trans| \leq b$  such that $\trans \models \varphi$. Furthermore, from the optimal satisfying assignment to the variables $\tau_{s,\inpval,s'}$ and $o_{s,\inpval}$, one can extract a transition system $\trans^*$ such that for every transition system $\trans$ with $|\trans| \leq b$ and $\trans \models \varphi$ it holds that $\val\trans{\gphi_1 \wedge \ldots \wedge \gphi_n} \leq \val{\trans^*}{\gphi_1 \wedge \ldots \wedge \gphi_n}$.
\end{appxthm}
\begin{proof}
The first part of the claim follows from the correctness of the classical bounded synthesis approach. More precisely, if the constraint system is satisfiable, then there exists a satisfying assignment, which, in particular satisfies the constraints asserting the existence of a transition system $\trans$ of size less than or equal to $b$, and the existence of valid annotation of the run graph of $\autA$ on $\trans$. If, on the other hand, there exists a transition system $\trans$ such that $|\trans| \leq b$ and $\trans \models \varphi$, then there exists a variable assignment $a$ consistent with $\trans$ that satisfies the constraints asserting the existence of valid annotation for $\autA$ and $\trans$. It remains to show that $a$ can be chosen in a way that satisfies the remaining hard constraints as well. To see that, notice that all the constraints for the annotations $\anotfg_j$ and $\anot_j$ can be satisfied (not necessarily in an optimal way) by setting all the variables $\anotfgbj_{s,q}$ and $\anotbj_{s,q}$ to $\falseval$. This completes the proof of the first statement.

Now, let $a^*$ be an optimal solution to the MaxSAT problem, and $\trans^*$ be the transition system extracted from $a^*$. Consider a transition system $\trans$ such that $|\trans| \leq b$ and $\trans \models \varphi$. Then, as we showed above, there exists a satisfying assignment $a$ consistent with $\trans$. Let $w^*$ be the sum of weights of soft clauses satisfied by $a^*$, and $w$ be the sum of weighs of soft clauses satisfied by $a$. Since $a^*$ is an optimal satisfying assignment we have that $w \leq w^*$. Thus, applying Lemma~\ref{lem:encoding-weights-appx} we obtain $\val\trans{\gphi_1 \wedge \ldots \wedge \gphi_n} \leq \val{\trans^*}{\gphi_1 \wedge \ldots \wedge \gphi_n}$, which concludes the proof of the second claim of the theorem.\qed
\end{proof}

\begin{appxprop}\label{prop:encoding-size-appx}
Let $\autA$ be a given co-B\"uchi automaton for $\varphi$, and for each $j \in \{1,\ldots,n\}$, let $\autB_j = \relaxfg(\gphi_j)$ be the universal B\"uchi automaton for $\gphi_j$ constructed as in Section~\ref{sec:automata-safety}, and let $\autA_j$ be a universal co-B\"uchi automaton for $\gfphi$. 
The constraint system for bound $b \in \nats$ has weights in $\mathcal{O}(n^2)$,
\[
\begin{array}{l}
\mathcal{O}\Big(
(b^2 + b \cdot |\outv|)\cdot 2^{|\inpv|} +  
b \cdot |Q|\cdot (1+ \log(b \cdot |Q|)) +\\ 
\phantom{\mathcal{O}(}\sum_{j=1}^n\big(b \cdot |Q_j| (1 + \log(b\cdot |Q_j|))\big)
+ \sum_{j=1}^n\big(b \cdot |\widehat Q_j| (1 + \log(b\cdot |\widehat Q_j|))\big)
\Big)
\end{array}\]
variables, and its size (before conversion to CNF) is
\[
\begin{array}{l}
\mathcal{O}\Big(
|Q|^2 \cdot b^2 \cdot 2^{|\inpv|} \cdot (d + \log(b\cdot |Q|)) + \\
\phantom{\mathcal{O}(}\sum_{j=1}^n\big(|Q_j|^2 \cdot b^2 \cdot 2^{|\inpv|} \cdot (d_j + r_j + \log(b\cdot |Q_j|))\big) +
\\
\phantom{\mathcal{O}(}\sum_{j=1}^n\big(|\widehat Q_j|^2 \cdot b^2 \cdot 2^{|\inpv|} \cdot (\widehat d_j + \log(b\cdot |\widehat Q_j|))\big)
\Big), 
\end{array}\]
\[
\begin{array}{lllllll}
\text{ where} &d &=& \max_{s,q,\inpval,q'}|\delta_{s,q,\inpval,q'}|,&
d_j &=& \max_{s,q,\inpval,q'}|\delta_{s,q,\inpval,q'}^j|,\\
&\widehat d_j &=& \max_{s,q,\inpval,q'}|\widehat \delta_{s,q,\inpval,q'}^j|, \text{ and } &
r_j &=& \max_{s,q,\inpval,q'}|\rej^j(s,q,q',\inpval)|.
\end{array}
\]
\end{appxprop}
\begin{proof}
The constraint system is defined in terms of the following variables:
\begin{itemize}
\item Boolean variables $\tau_{s,\inpval,s'}$ and $o_{s,\inpval}$ representing the transition system. The total number of these variables is $b^2 \cdot 2^{|\inpv|} + b \cdot |\outv| \cdot 2^{|\inpv|}$.
\item Boolean variables $\anotb_{s,q}$ and vectors of Boolean variables $\anotn_{s,q}$ representing the annotation $\anot$. The total number of bits is $b \cdot |Q| +b \cdot |Q|\cdot \log(b \cdot |Q|)$.
\item Boolean variables $\anotfgbj_{s,q}$ and vectors of Boolean variables $\anotfgnj_{s,q}$ representing the annotations $\anotfg_j$. The total number of bits is $\sum_{j=1}^n\big(b \cdot |Q_j| (1 + \log(b\cdot |Q_j|))\big)$.
\item Boolean variables $\anotbj_{s,q}$ and vectors of Boolean variables $\anotnj_{s,q}$ representing the annotation $\anot_j$. The total number of bits is $\sum_{j=1}^n\big(b \cdot |\widehat Q_j| (1 + \log(b\cdot |\widehat Q_j|))\big)$. 
\end{itemize}

The sum of the above quantities yields the total number of Boolean variables.

The constraint system consists of the following constraints:
\begin{itemize}
\item Constraints $C_\tau$ encoding input-enabledness, of size $b^2 \cdot 2^{|\inpv|}$.
\item Constraints $C_\anot$ for valid annotation $\anot$ of size $\mathcal{O}\big(|Q|^2 \cdot b^2 \cdot 2^{|\inpv|} \cdot (d + \log(b\cdot |Q|))\big)$.
\item Hard constraints $C_\anotfg^{j}$ for valid annotations $\anotfg_{j}$, of size \[\mathcal{O}\big(\sum_{j=1}^n\big(|Q_j|^2 \cdot b^2 \cdot 2^{|\inpv|} \cdot (d_j + r_j + \log(b\cdot |Q_j|))\big)\big).\]
\item Hard constraints $C_\anot^{j}$ for valid annotations $\anot_{j}$, of size \[\mathcal{O}\big(\sum_{j=1}^n\big(|\widehat Q_j|^2 \cdot b^2 \cdot 2^{|\inpv|} \cdot (\widehat d_j + \log(b\cdot |Q_j|))\big)\big).\]
\item Soft constraints $\soft_{\tiny\LTLglobally}^{j}$, $\soft_{\tiny\LTLfinally\LTLglobally}^{j}$ and $\soft_{\tiny\LTLglobally\LTLfinally}^{j}$ for valid annotations, of size
\[\mathcal{O}\big(\sum_{j=1}^n\big(|\log(b\cdot |Q_j|)\big)\big).\]
\end{itemize}

Summing up, we obtain the total size of the constraint system.\qed
\end{proof}

\newpage
\section{Generalizations of the Maximum Realizability Problem}\label{sec:generalizations}
\subsubsection{Maximum Realizability with Soft LTL Specifications.}\label{sec:soft-ltl}
The first generalization of the maximum realizability problem that we consider is the setting where the soft specifications can be arbitrary LTL formulas, and not just safety properties of specific form. More precisely, each soft specification $\softSpec$ is an LTL formula for which we are also given a vector $\Relax(\softSpec)$ of LTL formulas that defines the possible relaxations of $\softSpec$. Formally, $\Relax(\softSpec) = (\psi_{0},\ldots,\psi_{m})$, where $\psi_{0} = \softSpec$, and for every $0 \leq k < m$ it holds that $\trans \models \psi_{k}$ implies $\trans \models \psi_{k+1}$ for every transition system $\trans$. That is, $\psi_{0},\ldots,\psi_{m}$ are ordered according to strength.
In particular, if $\trans \models \softSpec$, then $\trans \models \psi_k$ for each $\psi_k$ in $\Relax(\softSpec)$.
For example, if $\softSpec = \LTLglobally p$ for some atomic proposition $p$, we can take $\Relax(\softSpec) = (\LTLglobally p,\LTLfinally\LTLglobally p,\LTLglobally\LTLfinally p)$.

As in Section~\ref{sec:quantitative-semantics} we define the value of $\softSpec$ for given $\Relax(\softSpec) = (\psi_{0},\ldots,\psi_{m})$ to be $\val\trans\softSpec = (v_1,\ldots,v_m)$, where $v_k = 1$ if $\trans \models \psi_k$, and $v_k = 0$ otherwise.

For a conjunction $\softSpec_1\wedge\ldots\wedge\softSpec_n$ of soft specifications with given $\Relax(\softSpec_j) = (\psi_{j,0},\ldots,\psi_{j,m})$ for each $j \in \{1,\ldots,n\}$, we define the value
$\val\trans{\softSpec_1 \wedge \ldots \wedge \softSpec_n} = \big(
\sum_{i=1}^n v_{i,m},
\ldots,
\sum_{i=1}^n v_{i,0}
\big),$ where
$\val\trans{\softSpec_i} = (v_{i,0},\ldots,v_{i,m})$ for $i \in \{1,\ldots,n\}$.

The maximum realizability problem asks for given LTL specification $\spec$ and soft LTL specifications 
$\softSpec_1,\ldots,\softSpec_n$ with given $\Relax(\softSpec_j) = (\psi_{j,0},\ldots,\psi_{j,m})$, to determine whether there exists a transition system $\trans$ such that $\trans\models\spec$, and if the answer is positive, to construct a transition system $\trans^*$ such that $\trans^* \models \spec$, and for every $\trans$ with $\trans\models\spec$ it holds that $\val\trans{\softSpec_1 \wedge \ldots \wedge \softSpec_n} \leq \val{\trans^*}{\softSpec_1 \wedge \ldots \wedge \softSpec_n}$.
The bounded maximum realizability problem is defined in the obvious way.

We can adapt the MaxSAT approach from Section~\ref{sec:maxsat} to solve the bounded maximum realizability problem in this setting as follows. 

First, for each $\psi_{j,k}$ in $\Relax(\softSpec_j)$, we construct a universal co-B\"uchi automaton $\autA_{j,k} = (Q_{j,k},q_0^{j,k},\delta_{j,k},F_{j,k})$ such that $\trans \in \lang{\autA_{j,k}}$ if and only if $\trans \models \psi_{j,k}$.
In the MaxSAT encoding, the hard constraints for the annotation $\anot_{j,k}$  are
\begin{align*}
C_{j,k}\defeq\bigwedge_{q,q' \in Q_{j,k}}\bigwedge_{s,s' \in S}\bigwedge_{\inpval \in \ialphabet}
\Bigg( &
\big(
\anotbjk_{s,q} \wedge
\delta^{j,k}_{s,q,\inpval,q'}\wedge 
\tau_{s,\inpval,s'}
\big) \rightarrow 
\succa_{\anot}^{j,k}(s,q,s',q',\inpval)
\Bigg).
\end{align*}
Generalizing the encoding, for each $j \in \{1,\ldots,n\}$ and $k \in \{0,\ldots,m\}$ we now have one soft constraint
$\soft_{j,k}\defeq  \bigvee_{l=0}^k \anotbjl_{s_0,q^{j,l}}$ with weight $n^k$.

\smallskip
\noindent 
{\bf Maximum Realizability with Priorities.}
In the definitions in Section~\ref{sec:max-realizability} and the paragraph above, all soft specifications have the same priority. Now, we extend the maximum realizability setting with priorities for the soft specifications given as part of the input to the problem.

\smallskip
\noindent
{\it Soft specifications with priority ordering.} We begin with a simple setting where soft specifications are simply ordered in decreasing priority, without assigning any numerical weight for each formula's preference. More specifically, we assume that the soft specifications $\softSpec_1,\ldots,\softSpec_n$ are ordered such that, for every $i \in \{1,\ldots,n\}$, we have that $\softSpec_i$ has higher priority than $\softSpec_j$ for all $j > i$. 

Now, given a vector $\Relax(\softSpec_j) = (\psi_{j,0},\ldots,\psi_{j,m})$ for $\softSpec_j$  we define the value of $\softSpec_j$ in a transition system $\trans$ to be the number of specifications in $\Relax(\softSpec_j)$ satisfied by $\trans$, i.e., $\val\trans{\softSpec_j} = | \{k \in \{0,\ldots,m\} \mid \trans\models \psi_{j,k}\}|$. The value of $\softSpec_1\wedge\ldots\wedge\softSpec_n$ is then defined as $\val\trans{\softSpec_1 \wedge \ldots \wedge \softSpec_n} = (\val\trans{\softSpec_1},\ldots,\val\trans{\softSpec_n})$. The values of transition systems are compared according to the lexicographic ordering of vectors in $\{0,\ldots,m+1\}^n$, giving priority of $\softSpec_i$ over $\softSpec_j$ for $i < j$.

The MaxSAT approach can be adapted for this value function in the same way as above. The difference is in the weights of the soft constraints for the annotations $\anot_{j,k}$: for each
$j \in \{1,\ldots, n\}$ and $k \in \{1,\ldots,m\}$ we have a soft constraint  $\soft_{j,k} \; \defeq \;  \bigvee_{l=0}^k \anotbjl_{s_0,q^{j,l}}$ with weight $w_{j,k}$, where 
$w_{j,k} = 1 $ if $j=n$ or $k < m$, and $w_{j,k} = \sum_{j' = j+1}^n \sum_{k = 1}^m w_{j',k} + 1$ otherwise.

\smallskip
\noindent
{\it Soft specifications with given weights.}
We also consider the weighted maximum realizability problem, in which, together with $\Relax(\softSpec)$ for each soft specification $\softSpec$, the user also provides numerical weights for the formulas in $\Relax(\softSpec)$. That is, for each $j \in \{1,\ldots,n\}$ and $k \in \{1,\ldots,m\}$ we are given a weight $w_{j,k}$ for $\psi_{j,k}$. These weights specify the priority of each of the soft specifications.

The MaxSAT formulation is then adapted to incorporate the given weights, by using them for the corresponding soft constraints. Namely, for each $j$ and $k$, the corresponding soft constraint $\soft_{j,k}  \; \defeq \;  \bigvee_{l=0}^k \anotbjl_{s_0,q^{j,l}}$ has weight $w_{j,k}$.

\comment{
\begin{theorem}\label{thm:optimal-bound-general}
Given an LTL specification $\spec$ and a soft specification $\gphi_1 \wedge\ldots \wedge\gphi_n$ together with a vector of formulas $\Relax(\softSpec_j) = (\psi_{j,0},\ldots,\psi_{j,m})$ for each $\softSpec_j$,
if there is a transition system $\trans$ with  $\trans \models \spec$, then there exists $\trans^*$ such that:
\begin{itemize}
\item $\val\trans{\softSpec_1 \wedge \ldots \wedge \softSpec_n} \leq \val{\trans^*}{\softSpec_1 \wedge \ldots \wedge \softSpec_n}$ for all $\trans$ with $\trans \models \spec$,
\item $\trans^* \models \spec$ and $|\trans^*| \leq (2^{b+\log b})!^2$, where
\end{itemize}
$b = \max\{|\subf{\spec\wedge\softSpec_1'\wedge\ldots\wedge\softSpec_n'}| \mid \softSpec_i' \in \Relax(\softSpec_i) \text{ for }i=1,\ldots,n\}$.
\end{theorem}
}

\begin{appxthm}\label{thm:optimal-bound-general-appx}
Given an LTL specification $\spec$ and a soft specifications $\gphi_1,\ldots,\gphi_n$ together with a vector of formulas $\Relax(\softSpec_j) = (\psi_{j,0},\ldots,\psi_{j,m})$ for each $\softSpec_j$,
if there is a transition system $\trans$ with  $\trans \models \spec$, then there exists $\trans^*$ such that:
\begin{itemize}
\item $\val\trans{\softSpec_1 \wedge \ldots \wedge \softSpec_n} \leq \val{\trans^*}{\softSpec_1 \wedge \ldots \wedge \softSpec_n}$ for all $\trans$ with $\trans \models \spec$, and
\item $\trans^* \models \spec$ and $|\trans^*| \leq (2^{b+\log b})!^2$,
\end{itemize}
where $b = \max\{|\subf{\spec\wedge\softSpec_1'\wedge\ldots\wedge\softSpec_n'}| \mid \softSpec_i' \in \Relax(\softSpec_i) \text{ for }i=1,\ldots,n\}$.
\end{appxthm}
\begin{proof}
The proof is a generalization of the proof of Theorem~\ref{thm:optimal-bound-safety}. First, we need to establish the analogue of Lemma~\ref{lem:value-as-ltl} for the general case. 

\begin{appxlemma}\label{lem:value-as-ltl-general-appx}
For every transition system $\trans$, soft specifications $\gphi_1,\ldots,\gphi_n$, and vector of formulas $\Relax(\softSpec_j) = (\psi_{j,0},\ldots,\psi_{j,m})$ for each $\softSpec_j$, if $\val\trans{\gphi_1 \wedge \ldots \wedge \gphi_n} = v$, then there exists an LTL formula $\psi_v$ such that $\trans \models \psi_v$ and the following hold:
\begin{itemize}
\item[(1)] $\psi_v = \softSpec_1'\wedge\ldots\wedge\softSpec_n'$, where $\softSpec_i' \in\Relax(\softSpec_i)  \text{ for }i=1,\ldots,n$, 
\item[(2)] for every $\trans'$, if $\trans' \models \psi_v$, then $\val{\trans'}{\gphi_1 \wedge \ldots \wedge \gphi_n} \geq v$.
\end{itemize}
\end{appxlemma}
The proof of Lemma~\ref{lem:value-as-ltl-general-appx} is analogous to the proof of Lemma~\ref{lem:value-as-ltl-appx}. Then, with the help of Lemma~\ref{lem:value-as-ltl-general-appx} we can establish the theorem in the same way as Theorem~\ref{thm:optimal-bound-safety}. \qed
\end{proof}

\end{document}